\documentclass[12pt]{article}
\usepackage{amsfonts,amssymb,amsbsy,amsmath,latexsym,natbib,graphicx}
\usepackage{psfrag,epsfig,epsf,dsfont,xcolor} 
\usepackage{algorithm}
\usepackage{algpseudocode}
\usepackage{subfigure}
\usepackage[margin=1in]{geometry}
\usepackage{comment}
\usepackage{booktabs}

\newtheorem{proposition}{Proposition}
\newtheorem{theorem}{Theorem}
\newtheorem{lemma}{Lemma}

\newcommand{\Prob}[1]{\mathbb{P}\left({#1}\right)}
\newcommand{\Probs}[2]{\mathbb{P}_{#1}\left({#2}\right)}
\newcommand{\Expect}[1]{\mathbb{E}\left[{#1}\right]}

\newcommand{\Expects}[2]{\mathbb{E}_{{#1}}\left[{#2}\right]}
\newcommand{\Var}[1]{\mathsf{Var}\left[{#1}\right]}

\newcommand{\SD}[1]{\mathsf{SD}\left[{#1}\right]}

\newcommand{\zvec}{\underline{z}}
\newcommand{\Zvec}{\underline{Z}}

\newcommand{\phat}{\widehat{p}}
\newcommand{\ellhat}{\widehat{\ell}}
\newcommand{\Phat}{\widehat{P}}
\newcommand{\Pbar}{\overline{P}}
\newcommand{\pbar}{\overline{p}}

\newcommand{\Stil}{\widetilde{S}}

\newcommand{\Wtil}{\widetilde{W}}

\newcommand{\md}{\mathsf{d}}

\newcommand{\cF}{\mathcal{F}}
\newcommand{\cX}{\mathcal{X}}
\newcommand{\cY}{\mathcal{Y}}

\newcommand{\Vrel}{V_{\mathrm{rel}}}
\newcommand{\ess}{\mathfrak{s}}
\newcommand{\ecc}{\mathfrak{c}}
\newcommand{\ind}[1]{\mathbb{I}\left({#1}\right)}

\algdef{SE}[DOWHILE]{Do}{doWhile}{\algorithmicdo}[1]{\algorithmicwhile\ #1}%

\graphicspath{../Graphics/}

\title{Robust, partially alive particle Metropolis-Hastings via The Frankenfilter}

\author{Chris Sherlock\footnote{Lancaster University}, Andrew Golightly\footnote{Durham University} and Anthony Lee\footnote{University of Bristol}}
\date{}

\begin{document}

\maketitle
\vspace{-1cm}

\begin{abstract}
  When a hidden Markov model permits the conditional likelihood of an observation given the hidden process to be zero, all particle simulations from one observation time to the next could produce zeros. If so, the filtering distribution cannot be estimated and the estimated parameter likelihood is zero. The Alive Particle Filter addresses this by simulating a random number of particles for each inter-observation interval, stopping after a target number of non-zero conditional likelihoods. For outlying observations or poor parameter values, a non-zero result can be extremely unlikely, and computational costs prohibitive.
We introduce the Frankenfilter, a principled, partially alive particle filter that targets a user-defined amount of success whilst fixing lower and upper bounds on the number of simulations. The Frankenfilter produces unbiased estimators of the likelihood, suitable for pseudo-marginal Metropolis--Hastings (PMMH). We demonstrate that PMMH with the Frankenfilter is more robust to outliers and mis-specified initial parameter values than PMMH using standard particle filters, and is typically at least 2-3 times more efficient. We also provide advice for choosing the amount of success. In the case of $n$ exact observations, this is particularly simple: target $n$ successes. 
\end{abstract}


\section{Introduction}
\label{sec.intro}

Consider a hidden Markov model for a $d_x$-dimensional hidden process, $\{X_t\}_{t\in \mathcal{T}}$, with index set $\mathcal{T}\subseteq [0,T]$ and state space $\cX$; $d_y$-dimensional observations $y_i\in \cY$, $i=1,\dots,n$, of the process are made at times $t_i\in \mathcal{T}$, with $t_n=T$. The likelihood for the parameter, $\theta$, is
\[
P(y_{1:n}|\theta)=\int f(y_{1:n}|x_{\mathcal{T}},\theta) ~\md \widetilde{P}(x_{\mathcal{T}}|\theta),
\]
where $\widetilde{P}$ is the law of $x_{\mathcal{T}}$, 
and the conditional likelihood is
\[
f(y_{1:n}|x_{\mathcal{T}},\theta)=\prod_{i=1}^n f(y_i|x_{t_i},\theta).
\]
Suppose that $f(y_i|X_{t_i}=x,\theta)=0$ if some criterion involving $X_{t_i}$ and $y_i$ is not satisfied. A case of particular note is when each conditional likelihood is an indicator function: $f(y_i|X_{t_i}=x):\cY\times \cX\rightarrow \{0,1\}$, for example with exact or partial but noiseless observations of a Markov jump process (MJP), which has a discrete state space, $\cX\subseteq \mathbb{Z}^d$; alternative scenarios include, for any Markov process, approximate Bayesian computation with a kernel such as $1(\|y_i-x_{t_i}\|\le \epsilon)$ \cite[e.g.,][]{Sisson2018ABC}. 

The usual collection of particle filters, including the bootstrap filter \cite[]{GSS1993} and the auxiliary particle filter \cite[]{PittShep1999} fix the number of particles (simulations) from one observation time, $t_{i-1}$ to the next, $t_i$. Thus, when $f(y_i|X_{t_i}=x)=0$ is possible, the conditional likelihood of every particle could be $0$; in this case, none of the particles at time $t_{i}$ can be resampled, and the particle filter `dies out'. Consequently, it is not possible to estimate the filtering distribution and, within a particle MCMC or SMC$^2$ scheme, the estimate of the likelihood is $0$. 

When a traditional particle filter is used within particle Metropolis--Hastings, the ideal number of particles can vary from one parameter value to another; for example, a parameter value $\theta$ that fits less well with the data than the posterior mode, $\widehat{\theta}$, will lead to simulations from the prior, $\widetilde{P}(\cdot|\theta)$, of the hidden process that fit less well with the data, and this can lead to more Monte Carlo variability in the estimate of the log likelihood for that particular simulation, necessitating a larger number of particles. Since in real scenarios (nearly) all models are wrong, even under $\widetilde{P}(\cdot|\widehat{\theta})$, deviations between the path of the true process (which produced $y_{1:n}$) and probable values under the model, or a mis-specified observation model, could lead to some of the observations being much less likely (under the model) than others, necessitating more simulations for the related inter-observation intervals. All of these points have motivated the development of the Frankenfilter.

\textbf{Notation and simplification}. We suppress dependence on the parameter, $\theta$, from the notation. In Sections 1-3, purely to simplify the presentation, we assume integer observation times ($t_i=i$, $i=1,\dots,n$; $t_n=n=T$) and that the transition kernel and observation process are time homogeneous. Here and throughout, $p_0(\cdot)$ denotes the prior for the hidden process at time $0$, and the one step transition kernel from time $t-1$ to time $t$, conditional on $X_{t-1}=x$, is $p(\cdot|x)$. The likelihood for the data conditional on $\theta$ becomes $P(y_{1:T})$.

\subsection{The alive particle filter and compromises in practice}

\cite{LeGlandOudjane2004,LeGlandOudjane2005} propose and analyse a particle filter which, for each inter-observation interval, simulates repeatedly until a prespecified amount of `success', $\ess$, has been achieved, thus producing a random number of particles. In the special case of a product of binary (indicator) likelihoods, simulations are started from a state chosen uniformly at random from the successful states at time $t$ , $\ess$ is the target number of successes ($1$s) and the estimate of the likelihood for $T$ observations is $\Phat_{LO}:=\prod_{t=1}^T \frac{\ess}{M_t}$, where $M_t$ is the number of simulations  required to achieve $\ess$ "successes" at time $t$.

Whilst $\Phat_{LO}$ is consistent for the true likelihood, $p$, as $\ess\to \infty$, it is biased and is, therefore, unsuitable for particle Metropolis--Hastings. \cite{AmrKun2011} and \cite{Alive2015} consider the case of a product of indicator function conditional likelihoods, and show that $\Expect{\Phat_{\mathrm{alive}}}=P(y_{1:T})$, where
\[
\Phat_{\mathrm{alive}}:=\prod_{t=1}^T \frac{\ess-1}{M_t-1}.
\]
Clearly, we must set $\ess\ge 2$. 
Subject to strong conditions, \cite{Alive2015} also shows that to control the variance of $\Phat_{\mathrm{alive}}/p$, $\ess$ should grow linearly with $T$.

Each $M_t$ is now unbounded and can be exceedingly large when the success probability is small, leading to a large, and potentially even infinite, expected runtime of the algorithm. The current ad hoc solution to this problem \cite[]{DrovMcC2016,DrovPetMcC2016} fixes a (high) hard threshold $m_+\ge \ess$ on $M_t$ such that if this threshold is reached, $M_t=m_+$, simulation halts and  $\Phat_{\mathrm{alive}}=0$ is returned. The motivation is that if $M_t=m_+>>\ess$, the simulated process does not fit with the data, so the associated parameter value should have low posterior mass and, hence, should be rejected.  

The `alive particle filter with a hard threshold', where success/failure is binary (as in  \cite{Alive2015,DrovMcC2016,DrovPetMcC2016}) is detailed in Algorithm \ref{alg.AliveHard}. Whether the Markov process is discrete in time or continuous in time, we refer to the set of states needed to propagate from $x_{t-1}$ to (and including) $x_t$ as the \emph{path} $x_{(t-1,t]}$ (strictly, $x_{(t-1,t]\cap \mathcal{T}}$). The algorithm counts a particle at time $t$ as successful if $f(y_t|x_t)=1$ and as a failure if $f(y_t|x_t)=0$. It starts by sampling a particle from the prior, $p_0$, and propagating it to time $1$, repeating until $\ess$ time-$1$ successes have been achieved. For each subsequent time, $t=2,\dots,T$, first an \emph{ancestor} is sampled uniformly from the pool of the first $\ess-1$ successful particles at time $t-1$; this particle is then propagated to time $t$ and, if successful it is added to the pool of particles that were successful at time $t$; again, this process is repeated until $\ess$ time-$t$ successes have been achieved.

Because of the hard threshold, Algorithm \ref{alg.AliveHard} leads to a biased estimator of the likelihood and biased inference for parameters within pseudo-marginal MCMC or SMC$^2$. The first simulation study, in Section \ref{sec:death}, demonstrates this bias in practice.

\begin{algorithm}
  \caption{Alive Filter with Hard Threshold}\label{alg.AliveHard}
    \begin{algorithmic}[1]
      \Procedure{AliveThresh}{$\ess,m_+$}\Comment{target \# successes, max \#trials}
      \For{$t$ in $1:T$}
      \State $m\gets 0$.
      \Do
      \State $m\gets m+1$.
      \If{$t=1$}
      \State Sample $x_{t-1}^{m}\sim p_0$; set $a_{t-1}^{m}=m$. 
      \Else
      \State $a_{t-1}^{m}\gets$ \textbf{DiscreteUniform}($\{1,\dots,\ess-1\}$).
      \EndIf
      \State Sample $\widetilde{x}_t^{m}\sim p(\cdot|x_{t-1}^{a_{t-1}^{m}})$; calculate $s^{m}\in\{0,1\}$ from $(x_t^{m},y_t)$.
      \If{$s^m=1$}
      \State{$x_t^{\sum_{j=1}^m s^j}=\widetilde{x}_t^m$}
      \EndIf
      \doWhile{$m<m_+$ \textbf{and} $\sum_{j=1}^{m} s^j< \ess$}
      \If{$m=m_+$} \label{HitThresh}
      \State \Return $\phat=0$. \Comment{$\phat_t=0\implies\prod_{s=1}^T \phat_s=0$}
      \Else
      \State $\phat_t\gets \frac{\ess-1}{m-1}$.
      \EndIf
      \EndFor
      \State \Return $\phat=\prod_{s=1}^T \phat_s$.
    \EndProcedure
    \end{algorithmic}
    \end{algorithm}

Inspired by \cite{Shelley1818} and motivated by the desire for the robustness and adaptivity of the fully alive filter and the practical necessity of bounding the amount of computational work, we present the Frankenfilter, a partially alive particle filter that creates an unbiased estimator of the likelihood. Although motivated by issues with indicator likelihoods and MJPs, the Frankenfilter may be employed for its robustness and adaptivity across the same set of Markov process and conditional likelihood regimes as the standard particle filter.

\subsection{Our contribution and the structure of this article}

The Frankenfilter corrects the alive particle filter with a hard threshold so that the likelihood estimator is unbiased, and generalises the applicability in several directions. In particular, the Frankenfilter allows: 
(i) specification of a minimum number of simulations, $m_-$, as well as a maximum, $m_+$;
(ii) the conditional likelihood to be any non-negative  weight rather than an indicator; in particular, this permits the use of informed transition proposals, so called `bridges' \cite[]{GoliWilk15,GoliSher2019};
  (iii) any non-negative measure of success, rather than an indicator; this permits the amount of success to depend on the conditional likelihood, for example.

We also derive tuning advice for the alive particle filter and the Frankenfilter. Simulation studies validate the tuning advice and compare PMMH using the Frankenfilter with an equivalent static particle filter, exploring the relative robustness and efficiency. 

The Frankenfilter can be applied to any  Markov process observed at a discrete set of times, where observations $y_t$ have a known conditional likelihood, $f(y_t|X_t=x_t)$, given the state of the hidden process, $X_t=x_t$, $t=1,\dots,T$. We will repeatedly refer to two cases of particular interest, where the conditional likelihood is an indicator function:
\begin{itemize}
\item \emph{Full, exact observations}: At each observation time, the full state vector is observed; \emph{i.e.}, we observe $X_{t}=x_t$, $t=1,\dots,T$.
\item \emph{Partial, exact observations}: At each observation time, a subset of the components of the state vector are observed. Let $X_t=(X^*_t,Z_t)$; we observe $X^*_t=x^*_t$, $t=1,\dots,T$.
\end{itemize}
Section \ref{sec.Frank} presents the Frankenfilter in several stages of increasing generality and proves the unbiasedness of the likelihood estimator. Section \ref{sec.tune} derives advice on tuning the alive particle filter and the Frankenfilter and Section \ref{sec.sims} presents the simulation studies. The article concludes in Section \ref{sec.Discussion} with a discussion.

\section{The Frankenfilter}
\label{sec.Frank}
To maximise the accessibility of the key ideas, we build up from the simplest possible scenario, dealt with by Algorithm \ref{BaseFrank}, generalise it in several ways (Algorithm \ref{OneStepFrank}), proving unbiasedness in the main text, and conclude with the full Frankenfilter (Algorithm \ref{FullFrank}).

\subsection{Full, exact observations}
\label{sec.base.algorithm}

We first present the simplest Frankenfilter, Algorithm \ref{BaseFrank}, which estimates $p=\Probs{\widetilde{P}}{X=x}$ for some $x$ by repeatedly sampling $X\sim \widetilde{P}$ and counting the number of successes ($X=x$).  Here, for the $j$th simulation, $w^j:=w(x^j)\in\{0,1\}$ indicates failure or success; for our purposes, $w(x^j)=1(x^j=x)$. We fix the desired number of successes, $\ess\ge 2$, and an upper bound on the number of simulations, $m_+\ge \ess$. The algorithm stops when either the number of successes reaches $\ess$ or $m_+$ iterations have been performed.

If $\ess=1$ then success could be achieved after the first simulation, leading to an estimate of $0/0$, which is undefined. Strictly, the algorithm is valid, producing an unbiased estimator of the success probability when $\ess\ge 2$. Typically, we will ask for $\ess>>2$; see Section \ref{sec.tune}.

In the case of full-vector, exact observations of a Markov jump process and a known initial condition, Algorithm \ref{BaseFrank} would be applied separately for the initial interval (simulating $X_1|X_0=x_0$) and then each inter-observation interval (simulating $X_t|X_{t-1}=x_{t-1}$).

\begin{algorithm}
\caption{Basic One-step Frankenfilter}\label{BaseFrank}
\begin{algorithmic}[1]
    \Procedure{FrankenFilterBaseOne}{$\ess,m_+$}\Comment{target \# successes, max trials}
      \State $m\gets 0$.
\Do
      \State $m\gets m+1$.
      \State Sample $x^m$ from $\widetilde{P}$; calculate associated weight, $w^m\in\{0,1\}$.
    \doWhile{$m< m_+$ \textbf{and} $\sum_{j=1}^{m}w^j< \ess$}
      \If{$\sum_{j=1}^m w^j<\ess$} \label{BaseEnd}
      \Return $\phat=\frac{1}{m}\sum_{j=1}^m w^j$.
      \Else $~$\Return $\phat=\frac{1}{m-1}\sum_{j=1}^{m-1} w^j$.
      \EndIf
    \EndProcedure
\end{algorithmic}
\end{algorithm}


Algorithm \ref{OneStepFrank} generalises Algorithm \ref{BaseFrank}, by allowing a pre-specified minimum number of simulations, for importance-based proposals and for an arbitrary (non-binary) measure of success. It also separates the measure of success from the particle weight.

As with Algorithm \ref{BaseFrank} it can be applied in the case of full-vector, exact observations, with realisations of independent estimators obtained for each inter-observation interval. The proof that the estimator produced by Algorithm \ref{OneStepFrank} is unbiased is included in the main text as the intuition is helpful for understanding the structure of the estimator.

\cite{Pathak1976} describes an unbiased estimator when samples of a real-valued variate are taken sequentially and each sample has an associated cost, which can only be ascertained after the sample has been taken. Sampling stops when a specific cost budget is reached. \cite{Kremers1987} considers Bernoulli variables and assumes a fixed, constant cost for samples, but allows upper and lower bounds on the number of samples taken. Combining these ideas, we sample $X^j\sim Q$, for some proposal distribution, $Q$, with density/mass function $q$, then set $W^j:= w(X^j)\ge 0$, where $w(\cdot)$ and $Q$ satisfy $\Expects{Q}{w(X)}=P(y)$. We relabel ``cost'' as ``amount of success'', set $S^j:= s(X^j)\ge 0$ and specify upper and lower bounds, $m_+$  and $m_-$, on the number of samples. 
Each simulation brings a certain amount of success, and the simulation stops if a success threshold $\ess$ is reached or $m_+$ simulations have been performed. We summarise the relevant notation:

\begin{itemize}
\setlength\itemsep{0em}
\item $w^j= w(x^j)$: \emph{weight}, or estimate of the quantity of interest, from the $j$th simulation. 
\item $s^j= s(x^j)$: \emph{amount of success} for the $j$th simulation.
\item $\ess$: required \emph{total success}; simulation ends when the total success equals or exceeds $\ess$. 
\item $m_-,m_+$: the \emph{minimum} and \emph{maximum} number of simulations.
\end{itemize}
In general, $w(x^j)=f(y|x^j)\widetilde{p}(x^j)/q(x^j)$, where $\widetilde{p}$ is the prior mass/density for $x^j$; thus $\Expects{Q}{w(X)}=P(y)$. When $Q=\widetilde{P}$ and $f(y|x^j)$ is binary, then $w^j$ is binary as in Algorithm \ref{BaseFrank}. Algorithm \ref{OneStepFrank} is valid for any non-negative $s(\cdot)$; we use $s(x)=f(y|x)/\sup_{x'\in \cX}f(y|x')$, which is binary when $f$ is binary and is justified more generally in Appendix \ref{app.general.success}.


\begin{algorithm}
  \caption{General One-step Frankenfilter}\label{OneStepFrank}
    \begin{algorithmic}[1]
      \Procedure{FrankenFilterGenOne}{$\ess,m_-,m_+$}\Comment{success threshold, min \& max}
      \State Sample $x^{1:m_-}$ iid from $Q$; calculate weights, $w^{1:m_-}$ and measures of success $s^{1:m_-}$.
      \State $m\gets m_-$.
      \While{$m<m_+$ \textbf{and} $\sum_{j=1}^m s^j< \ess$}
      \State $m\gets m+1$.
      \State Sample $x^m$ from $Q$; calculate $w^m$ and $s^m$.
      \EndWhile
      \If{$m=m_-$ \textbf{or} $\sum_{j=1}^m s^j<\ess$} \label{GenOneEnd}
      \Return $\phat=\frac{1}{m}\sum_{j=1}^m w^j$.
      \Else $~$\Return $\phat=\frac{1}{m-1}\sum_{j=1}^{m-1} w^j$.
      \EndIf
    \EndProcedure
    \end{algorithmic}
    \end{algorithm}

\textbf{On }$m_-$: if $\mathrm{ess} \sup S^j\ge \ess$, setting $m_- \ge 1$ avoids an undefined estimate.

\begin{proposition} 
\label{prop.GenOneUnbiased}
The random variable $\Phat$ returned by Algorithm \ref{OneStepFrank} satisfies
  \[
\Expect{\Phat}=\Expect{W}=P(y).
\]
\end{proposition}

\textbf{Proof}: Let $M$, $W^j$, $S^j$ be the random variables corresponding to $m$, $w^j$, $s^j$.

Let $K=0$ if $M=m_-$, $K=1$ if $M>m_-$ and $\sum_{j=1}^{m} S^j\ge \ess$ and $K=2$ if $\sum_{j=1}^{m_+} S^j<\ess$. 
The key concepts behind the proof of unbiasedness involve the exchangeability of the $(W^j,S^j)$ that is induced by the stopping conditions:

\begin{itemize}
\item If $K=0$, $W^{1:M}\equiv W^{1:m_-}$ are exchangeable because they were all sampled at the same time and all contributed to the fact that $\sum_{j=1}^{m_-} S^j \ge \ess$. 
\item If $K=2$, $w^{1:M}\equiv W^{1:m_+}$ are exchangeable because the only additional information we have on the $(S^j,W^j)$ pairs is that $\sum_{j=1}^{m_+} S^j<\ess$ (this implies that $\sum_{j=1}^m S^j<\ess$ for all $m<m_+$, so each of these earlier facts adds no extra information).
\item If $K=1$ then the $W^{1:M-1}$ are exchangeable since all we know about the $(S^j,W^j)$ pairs is that $\sum_{j=1}^{M-1}S^j<\ess$ and that adding $S^M$ brings the total to at least $\ess$.
\end{itemize}

Let $I=1$ if $K=1$ and $I=0$ if $K\in\{0,2\}$. Now,
\begin{align*}
  \Expect{\Phat|I=0}&=\Expect{\frac{1}{M}\sum_{j=1}^M W^j|I=0}=\Expect{W^1|I=0},
\end{align*}
where the second equality follows from the exchangeability of the $\{(W^j,S^j)\}_{j=1}^M$ given $I=0$. Similarly,
\[
\Expect{\Phat|I=1}=\Expect{\frac{1}{M-1}\sum_{j=1}^{M-1} W^j|I=1}=\Expect{W^1|I=1}.
\]
Combining, and then applying the tower law:
\[
\Expect{\Phat}=\Expect{\Expect{P|I}}=\Expect{\Expect{W^1|I}}=\Expect{W^1}.~~~\square
\]

\subsection{Partial and/or noisy observations}
\label{sec.general.alg}
The Frankenfilter (Algorithm \ref{FullFrank}) extends Algorithm \ref{OneStepFrank} to allow inference on multiple noisy and/or partial time-course observations. As in Algorithm \ref{alg.AliveHard}, each sampled particle is propagated from one observation time to the next, either in a single jump or via a path, and we allow for resampling. The notation $X_{(t-1,t]}$ and $x_{(t-1,t]}$ denotes the random path of the particle over $\mathcal{T}\cap (t-1,t]$ and its realisation. The relevant pool of time-$t$ ancestors depends on the outcome of the set of simulations from time $t-1$ to time $t$ and is specified in lines \ref{PoolStart} to \ref{PoolEnd}. If  $M=m_-$ ($K=0$ in the proof of Proposition \ref{prop.GenOneUnbiased}) or $\sum_{j=1}^{m_+} S^j<\ess$ ($K=2$) then the pool consists of all $m$ particles; otherwise ($K=1$) it is all except the particle that caused the cumulative success to hit the threshold.

Suppose that $X_0\sim p_0$ and $X_{(t-1,t]}|X_{t-1}=x\sim q(\cdot|x)$, an importance sampling (bridge) proposal. Let $a_{t-1}^j$ be the time $t-1$ ancestor of particle $x_t^j$. In general
\[
w^j=f(y_t|X_t=x_t^j)p(x_{(t-1,t]}^j|x_{t-1}^{a_{t-1}^j})/q(x_{(t-1,t]}^j|x_{t-1}^{a_{t-1}^j}).
\]

The algorithm fixes $\ess$ and $0\le m_-<m_+$; if $\ess\le \max_{t=1,\dots,T}\mathrm{ess}\sup(S_t^j)$, we insist $m_-\ge 1$.

\begin{algorithm}
  \caption{General Frankenfilter}\label{FullFrank}
    \begin{algorithmic}[1]
      \Procedure{FrankenFilter}{$\ess,m_-,m_+$}\Comment{success threshold, min \& max \#trials}
      \For{$t$ in $1:T$}
      \If{$t=1$}
      \State Simulate $x_0^{1:m_-}\sim p_0$; set $a_{t-1}^{1:m_-}\gets (1:m_-)$.
      \Else
      \State Sample $a_{t-1}^j\in\{1,\dots,m_{\mathrm{use}}\}$ with probability $\propto w_{\mathrm{use}}^{1:m_{\mathrm{use}}}$, $j=1,\dots,m_-$.
      \EndIf
      \State Sample $x_{(t-1,t)}^{j}\sim q(\cdot|x_{t-1}^{a_{t-1}^{j}})$; calculate $w^{j}$ and $s^{j}$, $j=1,\dots,m_-$.
      \State $m\gets m_-$.
      \While{$m<m_+$ \textbf{and} $\sum_{j=1}^m s^j< \ess$}
      \State $m\gets m+1$.
      \If{$t=1$}
      \State Sample $x_{t-1}^{m}\sim p_0$; set $a_{t-1}^{m}=m$.
      \Else
      \State Sample $a_{t-1}^{m}\in\{1,\dots,m_{\mathrm{use}}\}$ with probability $\propto w_{\mathrm{use}}^{1:m_{\mathrm{use}}}$.
      \EndIf
      \State Sample $x_{(t-1,t)}^{m}\sim q(\cdot|x_{t-1}^{a_{t-1}^{m}})$; calculate $w^{m}$ and $s^{m}$.
      \EndWhile
      \If{$m=m_-$ \textbf{or} $\sum_{j=1}^{m_t}s^j<\ess$} \label{PoolStart}
      \State $m_{\mathrm{use}}\gets m$
      \Else
      \State $m_{\mathrm{use}}\gets m-1$.
      \EndIf \label{PoolEnd}
      \State $w_{\mathrm{use}}^{1:m_{\mathrm{use}}}\gets w^{1:m_{use}}$; $\phat_t\gets\frac{1}{m_{\mathrm{use}}}\sum_{j=1}^{m_{\mathrm{use}}}w^j$
      \EndFor
      \State \Return $\phat=\prod_{t=1}^T \phat_t$.
    \EndProcedure
    \end{algorithmic}
    \end{algorithm}

\begin{theorem}
  \label{thrm.unbiased}
 The random variable $\Phat$ returned from Algorithm \ref{FullFrank} satisfies:
  \[  \Expect{\Phat(y_{1:T})}=P(y_{1:T}).
  \]
\end{theorem}

Theorem \ref{thrm.unbiased} is proved in Appendix \ref{sec.proof.unbiased.gen}.

\textbf{Generalisations}: one might specify different budgets for each inter-observation interval, $\ess_{1:T}$, with $\ess_t$ based on a prior estimate of $p_t:=P(y_t|y_{1:t-1})$, for example. This would still lead to an unbiased estimator of the likelihood. 
In Section \ref{sec.tune}, we show that with full exact observations it is reasonable to set $\ess_1=\ess_2=\dots=\ess_T=\ess$. Simulations in Section \ref{sec.sims} demonstrate that this is still a viable strategy, even with partial or noisy observations.

\section{Tuning}
\label{sec.tune}
Until recently, advice for tuning algorithms based on particle Metropolis--Hastings has been to choose the computational cost so that at some particular, representative parameter value, $\widehat{\theta}$, $\Var{\log \Phat(Y_{1:T})}\approx 1\text{--}2$ \cite[]{PitSilGioKoh2012,SheThiRobRos2015,DouPitDelKoh2015}.

\cite{SheVB2024} points out that in the case where $\Prob{\Phat=0}>0$, $\Var{\log \Phat(y_{1:T})}$ is not even defined. This particular case occurs with exact or partial observations of an MJP where the likelihood is zero when the state does not fit with the observation. More generally, \cite{SheVB2024} shows that tuning so that the \emph{relative variance},
\begin{equation}
    \label{eqn.define.Vrel.full}
\Vrel:=
\frac{\Var{\Phat(y_{1:T})}}{\Expect{\Phat(y_{1:T})}^2}=\frac{1}{P(y_{1:T})^2}\Var{\Phat(y_{1:T})}\approx 1\text{--}2,
\end{equation}
or a little above $1$, 
is a more robust approach as it can highlight potential large or infinite polynomial moments of $\Phat$. When $\Phat$ is well behaved and strictly positive, a first-order Taylor approximation shows that the two approaches are essentially equivalent.

All of the above advice assumes that  $\Var{\Phat}$ or $\Var{\log \Phat}$ is inversely proportional to the computational cost. In the case of a standard particle filter, the variance is approximately inversely proportional to the number of particles, which is proportional to the cost. For the Alive particle filter, Taylor expanding \eqref{eqn.VrelExact} (see later) shows that the variance is approximately inversely proportional to $\ess$, which is proportional to the expected cost.

Achieving too small a variance simply increases the computational cost slightly; too large a variance can negatively impact the mixing of the algorithm, potentially substantially.

Increasing $m_-$ can only lead to an estimator that is an average of more (or the same number) of terms, thereby reducing the Monte Carlo variance of $\Phat$ and making any recommended $\ess$ conservative. We will typically set $m_+$ sufficiently large that this upper threshold is met infrequently. Thus, some of our analysis is for the Alive particle filter, anticipating any resulting guidance will also apply to the Frankenfilter. We then examine the effect of $m_+$ on the overall variance and derive guidance on setting $m_+$.


\subsection{Tuning for full, exact observations with binary weights}
\label{sec.TuneExact}
Initially, we allow for different success thresholds for each observation, $\ess_1,\dots,\ess_T$. In the case of exact observations, the alive particle filter estimator of the transition probability from time $t-1$ to time $t$ is
\[
\Phat_t=\frac{\ess_t-1}{M_t-1},
\]
where $M_t$ is the number of attempts needed to obtain $\ess_t$ successes. 

In the case of full, exact observations,
\[
\Vrel=\frac{1}{\prod_{t=1}^Tp_t^2}\Expect{\prod_{t=1}^T\Phat_t^2}-1
=
\left\{\prod_{t=1}^T\frac{1}{p_t^2}\Expect{\Phat_t^2}\right\}-1.
\]
The following result, proved in Appendix \ref{sec.proof.PropBound}, enables us to progress the analysis of $\Vrel$.
\begin{proposition}
  \label{prop.bound}
Let $P$ be the random variable, with $\Expect{P}=p$, returned by Algorithm \ref{BaseFrank} with $m_+=\infty$, and define $E_{\mathrm{rel}}:=\Expect{P^2}/p^2$.
For $\ess=2$, 
\[
E_{\mathrm{rel}}=-\frac{\log p}{1-p},
\]
for $\ess=3$,
\[
 1+\frac{1-p}{3}\le E_{\mathrm{rel}}=\frac{2}{1-p}+\frac{2p\log p}{(1-p)^2}\le 2-p,
\]
and for $\ess \ge 4$, 
  \[
  1+\frac{1-(1+\frac{2}{\ess-3})p}{\ess-2}<E_{\mathrm{rel}}<1+\frac{1-p}{\ess-2}.
  \]
\end{proposition}

First, when $\ess=2$, we cannot control $E_{\mathrm{rel}}$ as $p\downarrow 0$. For $\ess \ge 3$, however, $E_{\mathrm{rel}}$ is bounded, and for large $\ess$, it behaves as $1+(1-p)/(\ess-2)$.

Using Proposition \ref{prop.bound}, then that each $\ess_t$ is typically large, and finally that in the cases where one might consider using the alive particle filter, $p_t$ is typically small,
\[
\left\{\prod_{t=1}^T\frac{1}{ p_t^2}\Expect{\Phat_t^2}\right\}
\approx
\prod_{t=1}^T\left\{1+\frac{1-p_t}{\ess_t-2}\right\}
\approx
\exp\left(\sum_{t=1}^T\frac{1- p_t}{\ess_t-2}\right)
\approx \exp\left(\sum_{t=1}^T\frac{1}{\ess_t-2}\right).
\]
Fixing this, the expected cost is $\propto \sum_{t=1}^T\ess_t$, which is minimised by setting $\ess_1=\dots =\ess_T=\ess$. Then,
\begin{equation}
  \label{eqn.VrelExact}
\Vrel
\approx
\exp\left(\frac{T}{\ess-2}\right)-1\approx 1\text{--}2,
\end{equation}
for good mixing. Note that, requiring $\Vrel=e-1$ suggests setting $\ess=2+T\approx T$.

It is worth emphasising that, unlike the case of the standard particle filter, this tuning advice does \emph{not} require a choice of a representative parameter value. The guidance is valid whatever the system and whatever the true (small) transition probabilities are.

\subsection{Tuning for partial, exact observations}
\label{sec.tune.partial}
We consider the alive particle filter in the case of partial, exact observations and multinomial resampling, and anticipate similar behaviour for the Frankenfilter provided $m_+$ is large. 

Let the full state at time $t$ be $X_t=(X^*_t,Z_t)$, where $X^*_t$ is observed precisely and $Z_t$ is unobserved. 
At time $t-1$ there are $\ess-1$ particles (with values $z_{t-1}^1,\dots z_{t-1}^{\ess-1}$). We sample $Z_{t-1}$ uniformly at random from these and propagate $X_{t-1}=(x^*_{t-1},Z_{t-1})$ forwards until time $t$, registering the simulation as a success if $X^*_t=x^*_t$, and a failure otherwise. This is repeated until we achieve $\ess$ successes at time $t$. The first $\ess-1$ successes provide $\zvec_t:=(z_t^1,\dots,z_t^{\ess-1})$. 

Define $p_{t}(z)=\Prob{X_t=x_t|X_{t-1}=(x^*_{t-1},z)}$ and for $j=1,2,\dots,\ess-1$, let  $\pbar_t=\frac{1}{\ess-1}\sum_{j=1}^{\ess-1}p_{t}(z_{t-1}^j)$; for the random variable, we write $\Pbar_t=\frac{1}{\ess-1}\sum_{j=1}^{\ess-1}p_{t}(Z_{t-1}^j)$.

For $t\in \{1,\dots,T\}$, let $\cF_t$ be the results of all simulations up to and including time $t$ (the filtration at time $t$). All the relevant information from $\cF_{t-1}$ that will inform the simulated values and estimated success probability at time $t$, as well as the evolution beyond time $t$ is contained in $\Zvec_{t}$.

There are two special aspects of our scenario that facilitate helpful analytical insight.

\textbf{S1}: Because we are choosing the ancestor of the $j$th simulation uniformly at random, when simulating $\Phat_t$, the distribution of the number of trials until $\ess$ successes is indistinguishable from the distribution of the number of trials if each were a simple Bernoulli trial with success probability $\Pbar_t$.

\textbf{S2}: Conditional on $\cF_{t-2}$, $\Phat_{t-1}$ and $(\Pbar_t,\dots, \Pbar_T)$ are independent as the forward simulations from time $t-2$ to time $t-1$ are mutually independent. In particular, the number of failures, $F_{t-1}:=|\{i:{X^{*i}_{t-1}}\ne x^*_{t-1}\}|$ is independent of the state values at the successes, and, therefore, of $\zvec_{t-1}$. Now, $\Phat_{t-1}=(\ess-1)/(\ess+F_{t-1}-1)$, whereas $\Pbar_t$ is a deterministic function of $\Zvec_{t-1}$ and the subsequent evolution only depends on $\Zvec_{t-1}$, not $F_{t-1}$.

As before, $\Vrel$ is as defined in \eqref{eqn.define.Vrel.full};
now, however, we only know that $p=\Expect{\prod_{t=1}^T \Phat_t}$. 
Using S1 and Proposition \ref{prop.bound} we have
\begin{equation}
\Expect{\Phat_t^2|\cF_{t-1}}
=
\Expect{\Phat_t^2|\Zvec_{t-1}}
=
\Expect{\Phat_t^2|\Pbar_t}
\le
\left(1+\frac{1}{\ess-2}\right)\Pbar_t^2,~~~t=1,\dots,T,
\label{eqn.usefulInterim}
\end{equation}
and, again from Proposition \ref{prop.bound}, when $\ess$ is large and each $p_t$ is small, we expect this bound to be close to an equality. Theorem \ref{thrm.partialVariance}, below, is proved in Appendix \ref{sec.prove.partialVariance}.

\begin{theorem}
\label{thrm.partialVariance}
In the case of the alive particle filter on partial, exact observations with the same target number of successes, $\ess$, for each observation, 
\begin{equation}
  \label{eqn.Vrel.gen}
\Vrel
\le 
\exp\left\{\frac{T}{\ess-2}\right\}
\frac{\Expect{\prod_{t=1}^T\Pbar_t^2}}{\Expect{\prod_{t=1}^T\Pbar_t}^2}-1
=
\exp\left\{\frac{T}{\ess-2}\right\}
\frac{1}{P(y_{1:T})^2}\Expect{\prod_{t=1}^T\Pbar_t^2}-1
.
\end{equation}
\end{theorem}

When $\ess$ is large, we expect \eqref{eqn.Vrel.gen} to be close to an equality. Rewriting this gives, 
\[
\Vrel
\approx
\exp\left\{\frac{T}{\ess-2}\right\}-1+\exp\left\{\frac{T}{\ess-2}\right\}\frac{\Var{\prod_{t=1}^T\Pbar_t}}{P(y_{1:T})^2}.
\]
We see that for a given $\ess$, the relative variance is larger than in the case of exact observations. So, to achieve a given target relative variance, $\ess$ must increase. Since each $\Pbar_t$ is an average of $\ess-1$ (albeit correlated) probabilities, when starting from $\ess \approx T$, we anticipate $\Var{\prod_{t=1}^T \Pbar_t}$  decreasing roughly in proportion to $1/\ess$ as $\ess$ increases. Appendix \ref{sec.more.partial} provides further analysis in the case of very large $\ess$.

\subsection{Choice of $m_+$}
\label{sec.choose.mplus}
We now show that choosing $m_+$ such that $m_+p =\kappa \ess$ for some moderately sized $\kappa$, such as $\kappa=10$, is sufficient to control the additional extra variance that comes from the introduction of $m_+$ into the alive particle filter. We restrict attention to the most obvious application of Algorithm \ref{OneStepFrank}, where $s^j=w^j/c$, $j=1,2,\dots$, for some $c>0$. Propositions \ref{prop.GenOneStepBound} and \ref{prop.BoundProbMissS} are proved in Appendix \ref{sec.proof.PropBound}.

\begin{proposition}
  \label{prop.GenOneStepBound}
Let $c>0$ and let $P$ be the random variable, with $\Expect{P}=p$, returned by Algorithm \ref{OneStepFrank} in the special case that $S^j= \frac{1}{c}W^j$, $j=1,2,\dots$. If $0\le W^j\le w_*$, $j=1,2,\dots$ and $c\ess >2w_*$,
\[
\frac{1}{p^2}\Expect{P^2}
\le
1+\frac{w_*}{c\ess-2w_*}+
\frac{c\ess}{m_+ p}\Prob{\sum_{j=1}^{m_+}S^j<\ess}.
\]
\end{proposition}
The term $1+w_*/(c\ess-2 w_*)$ is a slight loosening of the upper bound in Proposition \ref{prop.bound}, while $\Prob{\sum_{j=1}^{m_+}S^j<\ess}$ is the probability that the bound $m_+$ will make a difference to the output. 
The result gives an upper bound on the relative variance, generalising the approximation in \eqref{eqn.VrelExact}:
\begin{equation}
  \label{eqn.VrelGenOne}
\Vrel
\le
\exp\left(\frac{Tw_*}{c\ess-2w_*}+\sum_{t=1}^T\frac{c\ess}{m_+ p_t}\Prob{\sum_{j=1}^{m_+} S^j_t<\ess}\right)-1.
\end{equation}
The additional variance comes from the second term in the sum, which we, therefore, wish to keep small but not negligible, since, if it were negligible, the Frankenfilter would perform nearly as much work as the alive filter. 
Bernstein's inequality gives:

\begin{proposition}
\label{prop.BoundProbMissS}
If $m_+ p =\kappa c\ess$ for some $\kappa \ge 7/4$ then
\[
\Prob{\sum_{j=1}^{m_+}S^j< \ess}<
\exp\left\{-\frac{3\left(\kappa-\frac{7}{4}\right) c\ess}{8w_*}\right\}.
\]
\end{proposition}

Combining with Proposition \ref{prop.GenOneStepBound}, we obtain
\[
\frac{1}{p^2}\Expect{P^2}
<
1+\frac{w_*}{c\ess-2w_*}+
\frac{1}{\kappa}\exp\left\{-\frac{3(\kappa-\frac{7}{4}) c\ess}{8w_*}\right\}.
\]
In the typical regime where $c\ess/w_*>>1$, the bound on the probability that $m_+$ actually plays a role in the form of $P$, is controlled with relatively small values of $\kappa$, so this is sufficient to control the whole final term.

In practice, we do not know each $p_t$; however, during a trial run of pseudo-marginal Metropolis--Hastings (PMMH), it is straightforward to output the fraction of times the Frankenfilter finishes sampling without achieving total success $\ess$ (see Table \ref{table.LVvaryMpl}, for example); $m_+$ can be set so that this is small but not negligible. Alternatively, a  very approximate lower bound on each $p_t$ could be used to set a reasonable $m_+$ directly; in other parts of Section  \ref{sec.sims}  we use a preliminary run of the alive particle filter at the true parameter value. 

\section{Simulations}
\label{sec.sims}

In this section we empirically validate the tuning advice of Section \ref{sec.tune} and evaluate the performance of the Frankenfilter inside PMMH by considering five applications based upon MJPs, their discretisation and their approximation through a stochastic differential equation (SDE).

We allow the index set, $\mathcal{T}$, of $X_t$ to be discrete or continuous. To fix additional notation to that in Section \ref{sec.intro}, for MJPs, $\mathcal{X}\subseteq \mathbb{Z}^{d_x}$ and (typically) $\mathcal{Y}\subseteq\mathbb{Z}^{d_y}$; for SDEs, $\mathcal{X}\subseteq \mathbb{R}^{d_x}$ and (typically) $\mathcal{Y}\subseteq \mathbb{R}^{d_y}$. In our MJP examples, $y_{i}=F^\top x_{t_i}$ for some constant $d_x \times d_y$ matrix $F$ and in our SDE example, $Y_{i}\sim \mathsf{Gamma}(\alpha,\alpha/x_{t_i})$ for some $\alpha>0$. We consider three specific Markov process models as follows. In the applications in Sections \ref{sec:death} and \ref{sec:mule}, the process $(X_t)_{t\geq 0}$ is modelled as an MJP. Given the state of the system at time $t$ as $x_t$, the process transitions to time $t+\tau$ via
\begin{equation}\label{eqn.MJP}
X_{t+\tau}=x_t+\sum_{i=1}^r S^{i}\mathcal{P}_i\left(\int_{t}^{t+\tau}h_i(x_{t'},\theta)dt' \right),
\end{equation}
where $\mathcal{P}_i(t)$, $i=1,\ldots,r$ are independent, unit rate Poisson processes, $S^i$ is the $i$th row of a $d_x \times r$ scaling matrix $S$ whose elements are non-negative integers and the hazard function $h(x_{t}):=(h_{1}(x_t,\theta),\ldots,h_r(x_t,\theta))^\top$ is a vector of event rates or hazards parameterised by $\theta$. MJPs are routinely used to model reaction networks consisting of $r$ reaction channels; in this case, $X_t$ represents a vector of specie counts, $\theta$ is a vector of reaction rate constants and $S$ is known as the stoichiometry matrix, which describes the effect of each reaction on each component of $X_t$ \citep[see e.g.][]{Wilkinson06}. Exact simulation of the MJP is straightforward via Gillespie's direct method \citep{Gillespie77}. For the applications in Sections~\ref{sec:dim} and \ref{sec:lv} where simulating the MJP exactly is likely to be computationally preclusive, we consider a tau-leap approximation \citep[][]{Gillespie01}. In this case, the state $x_t$ is assumed constant over the time interval $[t,t+\tau)$ so that (\ref{eqn.MJP}) becomes       
\begin{equation}\label{eqn.tauleap}
X_{t+\tau}=x_t+\sum_{i=1}^r S^{i}\mathcal{P}_i\left(h_i(x_{t},\theta)\tau \right).
\end{equation}
Tau-leaping permits an adaptive time-step $\tau$ \citep[see e.g.][]{Cao07}. In what follows, for simplicity, we assume a constant $\tau$. In Section \ref{sec.LVSDE}, $(X_t)_{t\ge 0}$ is the solution of a stochastic differential equation (SDE),
\[
\md X_t = a(X_t)\md t + b(X_t)\md W_t,~~~X_0=x_0,
\]
where $a:\mathbb{R}^{d_x}\to \mathbb{R}^{d_x}$ is the drift, $b:\mathbb{R}^{d_x}\to \mathbb{R}^{d_x\times d_x}$ is the volatility and $W_t$ is a $d_x$-dimensional vector of independent Brownian motions.

For PMMH runs, all filters and simulators were coded in \texttt{C++} and implemented in R via the package \texttt{Rcpp}. All runs were performed on a desktop computer with an Intel Core i7-10700 processor and a 2.90GHz clock speed. For PMMH, a Gaussian random walk proposal was applied to the natural logarithm of the parameter (vector), with innovation variance $\gamma \Var{\log\theta|y_{1:n}}$, and the scaling $\gamma$ chosen to meet a desired acceptance rate \citep[see Table~1 of][]{schmon2021}. Code for the SDE experiments was written in R and run on an Intel Core Ultra 5 235U (up to 4.9 GHz). Code to implement the Frankenfilter and competing methods for each application is available from https://github.com/AndyGolightly/Frankenfilter .

Throughout, we abbreviate the bootstrap particle filter to BSPF, the Frankenfilter to FF and the alive particle filter to APF. When tabulating results, for the BSPF, $m_+$ represents the actual number of particles used and for the APF it represents the hard threshold, after which the algorithm is terminated and returns an estimated likelihood of 0. Effective sample size is abbreviated to ESS and effective samples per second to ESS/s.

\subsection{Pure death process}\label{sec:death}
We consider an MJP representation of a pure death process with hazard $h(x_t,\theta)=\theta x_t$ corresponding to a single reaction channel of the form $\mathcal{X} \longrightarrow \emptyset$, with stoichiometry matrix $S=-1$. For the model, the transition probability of $X_t$ (given the initial state $x_0$) is tractable and given by 
\[
\Prob{X_t=x}=\mathsf{Bin}\left(x;x_0,e^{-\theta t}\right).
\]

We simulated two synthetic data sets on the time interval $[0,50]$ with $x_0=100$ and $\theta=\theta_{\text{true}}=0.01$. The first data set, denoted D50, consists of observations at integer times. The second data set, denoted D50mod, is as D50, albeit with the last two observations replaced by the lower $0.01\%$ quantiles of $X_{49}|X_{48}$ and $X_{50}|X_{49}$ respectively so $(x_{48},x_{49},x_{50})=(58, 57, 57)$ rather than $(58, 53, 48)$. We assumed that $\theta \sim \mathsf{Gamma}(10,1000)$ \emph{a priori} and ran PMMH for $50K$ iterations with: the BSPF, FF and APF using the settings in Table~\ref{tab:tabD}. The number of particles for the BSPF ($m_{+}$) was chosen by following the practical advice of \cite{SheThiRobRos2015}. For the FF, we set $m_{-}=0$ and chose $m_{+}$ to be the number of particles used in the BSPF. To assess the posterior accuracy of each filter-driven PMMH scheme, we also ran a Metropolis-Hastings (MH) scheme using the \emph{exact} likelihood, evaluated using the tractable transition probabilities.      

\begin{table}[ht]
\centering
\small
	\begin{tabular}{@{}ll rrr rrrr@{}}
         \toprule
Data set & Filter  & $\ess$ &  $m_{+}$ & CPU (s) & ESS & ESS/s & Mean (SD) & SE  \\
\midrule
D50 & Direct  & -&  -&   12& 8792& 732.7&1.088 (0.149) & 0.0016 \\
    &BSPF & -  &400   & 527& 2552&  4.8& 1.084 (0.148) & 0.0029\\
&APF      & 50& 400   & 226& 1180&  5.2& 0.872 (0.200) & 0.0058\\
&         & 50& $10^4$& 423& 3263&  7.7& 1.107 (0.138) & 0.0024\\
&FF       & 50& 400   & 221& 2227& 10.1& 1.085 (0.147) & 0.0031\\
\midrule
D50mod &Direct   & -& -& 12& 9542& 795.2&       1.276 (0.163) & 0.0017\\
       &BSPF     & -& $10^4$& 11824&3415& 0.3 & 1.280 (0.163) & 0.0028\\
&APF      & 50& $10^4$&  745&   58      & 0.1 & 2.142 (0.081) & 0.0107 \\
&         & 50& $10^6$& 5899& 3050      & 0.5 & 1.263 (0.160) & 0.0029\\       
&FF       & 50& $10^4$&  799& 2458      & 3.1 & 1.279 (0.160) & 0.0032\\
\bottomrule
\end{tabular}
	\caption{Death model. Filter used inside PMMH, required total success $\ess$, maximum number of simulations $m_{+}$, CPU time (in seconds), ESS, ESS/s, estimated posterior mean of $\theta/\theta_{\text{true}}$ (and standard deviation in parentheses) and standard error (SE) of this posterior mean. All results are based on $50K$ iterations of PMMH.}\label{tab:tabD}
\end{table}

\begin{figure}[ht]
\centering
\includegraphics[width=16.0cm,height=6.0cm]{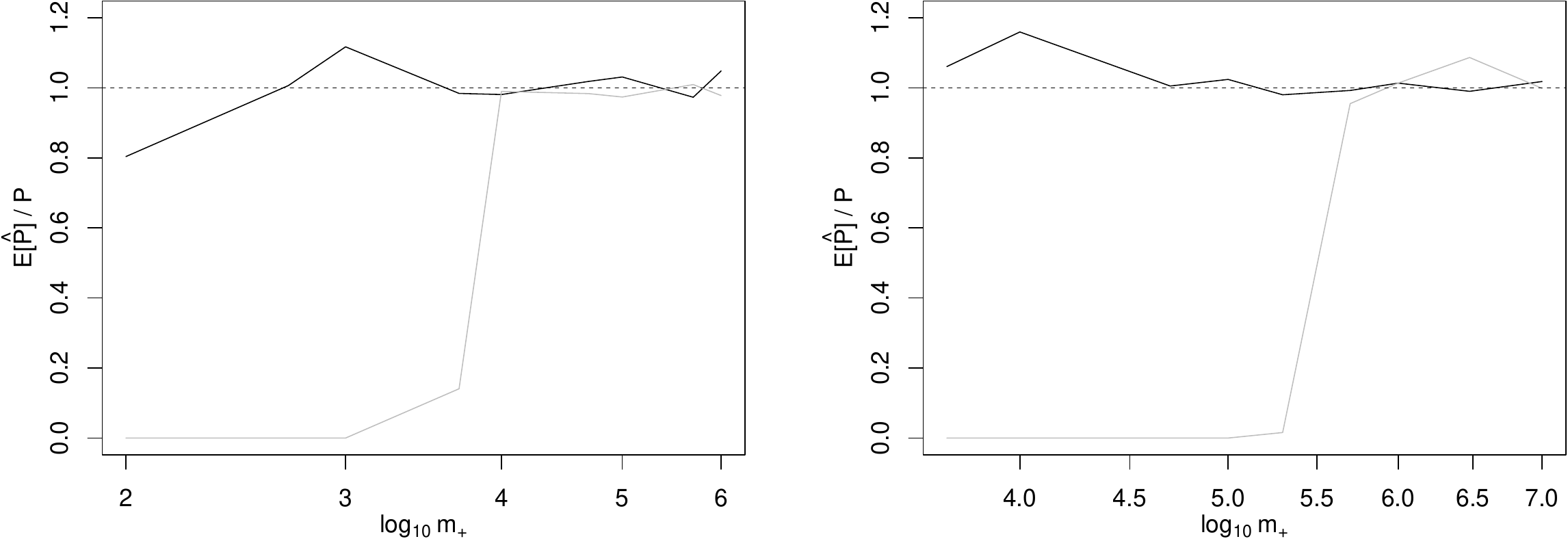}
\caption{Death model. Ratio of the expectation of the estimator of the likelihood and exact likelihood at $\theta=0.01$, based on the Frankenfilter (black lines) and alive particle filter (grey lines) using D50 (left panel) and D50mod (right panel).}
\label{fig:figD}
\end{figure}

It is evident from Table~\ref{tab:tabD} that for D50, there is little difference in overall efficiency between PMMH schemes, as measured by effective sample size per second (ESS/s). FF outperforms BSPF and APF by factors of 2 and 1.3-2, respectively. However, failure to choose a sufficiently large number of simulations $m_+$ in APF results in biased inferences for $\theta$, due to the use of realisations of a biased estimator of likelihood in PMMH (see Figure~\ref{fig:figD} demonstrating this bias at the ground truth $\theta$). A $3\widehat{\sigma}/\sqrt{\mathrm{ESS}}$ interval for the posterior expectation of $\theta/\theta_{\text{true}}$ using the direct method is $(1.083, 1.093)$ which does not overlap with the same interval under APF, even for $m^+=10^4$. Both BSPF and FF on the other hand, give inferences consistent with the direct method. For D50mod, the outlying final two observations have transition probabilities (likelihoods), at the ground truth $\theta$, of $2.3 \times 10^{-4}$ and $1.5\times 10^{-4}$, necessitating a much larger value of $m_{+}$ for BSPF, and, in order to obtain accurate inferences, for APF. Consequently, FF outperforms BSPF and APF by factors of 10 and 6-30 respectively. The particularly small ESS obtained for APF with $m_{+}=10^4$ was due to the failure of the algorithm to generate a non-zero estimate of the likelihood for most iterations. Whereas $3\widehat{\sigma}/\sqrt{ESS}$ intervals for the posterior expectation using BSPF and FF contain the estimate from the direct method, the interval from APF with $m_+=10^6$ only just intersects the interval from the direct method, and the interval using the APF with $m_+=10^4$ does not intersect it at all.

\subsection{Protein dimerisation}\label{sec:dim}
A simple model of protein dimerisation involves two reactions of the form $2\mathcal{P} \longrightarrow \mathcal{P}_2$ and $\mathcal{P}_2 \longrightarrow 2\mathcal{P}$, where $\mathcal{P}$ represents protein and $\mathcal{P}_2$ represents dimers of the protein. The stoichiometry matrix and hazard vector are given by
\[
S=\begin{pmatrix} -2 & 1\\ 2 & -1 \end{pmatrix}, \quad h(x_t,\theta)=(\theta_1 x_{1,t}(x_{1,t}-1)/2, \theta_2 x_{2,t})^\top
\]
where $x_{1,t}$ and $x_{2,t}$ denote the respective numbers of proteins and dimers, at time $t$. Using the MJP representation (\ref{eqn.MJP}), we simulated several data sets using the settings in Table~\ref{tab:tabPD}; since $x_1+2x_2$ is conserved, species numbers in the P*0b data sets are typically ten times those in the P*0a data sets, and so transition probabilities are typically reduced by a factor of $10$. In what follows, inference is performed on the tau-leaping model in (\ref{eqn.tauleap}) with $\tau=0.1$. 


We used the Frankenfilter inside PMMH, with $m_{-}=0$. Following the advice in Section \ref{sec.choose.mplus}, we set $m_{+}=10\ \ess/\min_j \hat{P}_j$, where $\hat{P}_j$ is the estimated transition probability at time $j$, obtained by running the alive particle filter with $\theta$ fixed at the ground truth and $\ess=5T$. Hence, $m_{+}$ is chosen conservatively, so that the tuning advice for $\ess$ in Section~\ref{sec.TuneExact} should be appropriate. We specified $\theta_1\sim \mathsf{Gamma}(2,500)$ and $\theta_2 \sim \mathsf{Gamma}(2,2)$ \emph{a priori} and ran PMMH for $50K$ iterations for each data set. We empirically validate the tuning advice of Section \ref{sec.TuneExact} by reporting overall efficiency (ESS/s) against various choices of the required total success $\ess$ in Figure~\ref{fig:figPD}. Since there are two parameter chains, we take ESS as multivariate ESS \citep[see e.g.][]{vats19}. Setting $\ess=T$ is clearly a viable strategy, with this choice giving optimal or close to optimal (at least $85 \%$ of maximum efficiency) across all data sets. Adopting a slightly more conservative strategy by requiring $V_{rel}=1$ suggests $\ess=16, 45$ and $74$ for data sets with $T=10, 30$ and $50$ observations. For data sets P10a, P30a and P50a, this would result in efficiency factors (efficiencies relative to the maximum obtained across the different $\ess$ values) of $77\%$, $79\%$ and $71\%$, respectively. For data sets P10b, P30b and P50b, the relative efficiency factors are $91\%$, $87\%$ and $89\%$, respectively. Data sets P10a, P30a and P50a have relatively large typical transition probabilities at the true $\theta$. With the approximation of $\sum_{t=1}^T 1-p_t\approx T$ made in deriving the tuning advice, we might expect the optimal choice of $\ess$ to be slightly lower than $T$; however, this is only clearly apparent for the P50a data set, where the posterior for $\theta$ is tighter about the truth.         

\begin{table}[ht]
\centering
\small
	\begin{tabular}{@{}l rrr rrr@{}}
         \toprule
Data set & $\theta$  & $x_0$ &  $\Delta t$ & $T_{max}$ & $T$ & $\frac{1}{T}\sum_{t=1}^{T}p_t$  \\
\midrule
P10a & (0.00332, 0.2) & (20, 1) & 1 & 10 & 10 & 0.34 \\
P10b & (0.00332, 0.2) & (200, 10) & 1 & 10 & 10 & 0.07 \\
\midrule
P30a & (0.00332, 0.2) & (20, 1) & 1 & 30 & 30 & 0.35 \\
P30b & (0.00332, 0.2) & (200, 10) & 1 & 30 & 30 & 0.08 \\
\midrule
P50a& (0.00332, 0.2) & (20, 1) & 1 & 50 & 50 & 0.34 \\
P50b& (0.00332, 0.2) & (200, 10) & 1 & 50 & 50 & 0.07 \\
\bottomrule
\end{tabular}
	\caption{Protein dimerisation model. Synthetic data set name, ground truth parameter values $\theta$, initial condition $x_0$, inter-observation time $\Delta t$, final time $T_{max}$, number of observations $T$ and average transition probability $p_t:=\Prob{X_t=x_t|X_{t-1}=x_{t-1}}$ computed at the ground truth $\theta$.} \label{tab:tabPD}
\end{table}

\begin{figure}[ht]
\centering
\includegraphics[width=16.0cm,height=5.0cm]{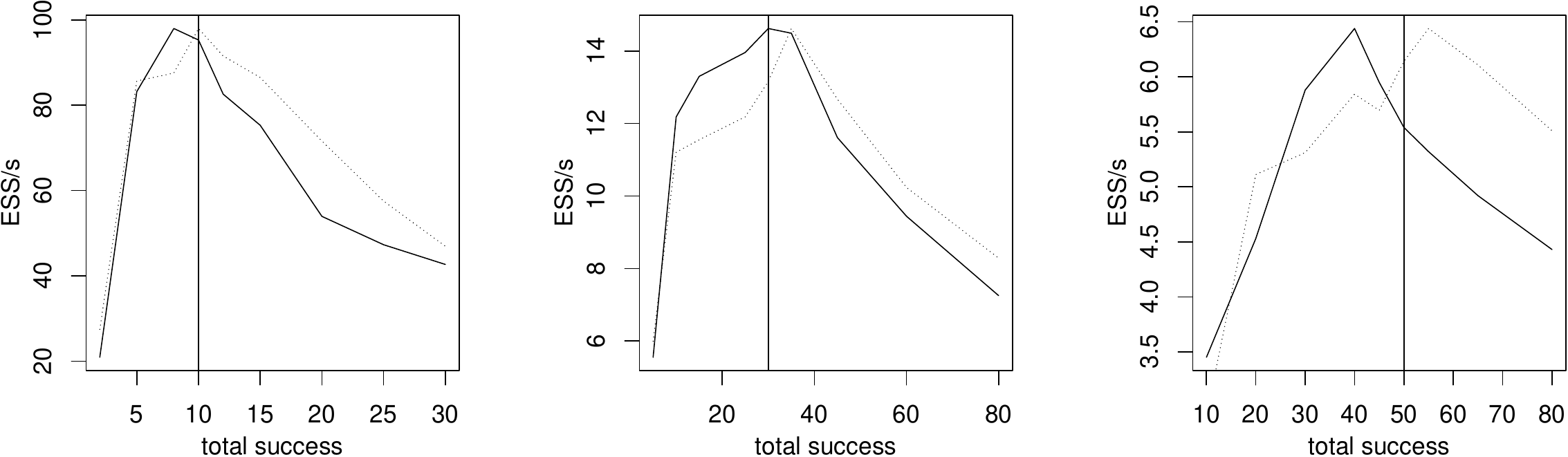}
\caption{Protein dimerisation model. ESS/s versus required total success $\ess$, from the output of $50K$ iterations of PMMH, using data sets with 10 observations (left panel), 30 observations (middle panel), 50 observations (right panel). Solid lines correspond to data sets P10a, P30a and P50a whereas dotted lines correspond to data sets P10b, P30b and P50b (with ESS/s scaled to give the same maximum within observation regime). The value $\ess=T$ is indicated by a vertical line.}
\label{fig:figPD}
\end{figure}

\subsection{Lotka--Volterra tau leap}\label{sec:lv}
The Lotka-Volterra model of predator-prey interaction is routinely used to validate and compare inference schemes \citep[see e.g.][]{BWK08}. The associated reaction network for prey ($\mathcal{X}_1$) and predators ($\mathcal{X}_2$) has three reactions of the form  $\mathcal{X}_1 \longrightarrow 2\mathcal{X}_1$, $\mathcal{X}_1 + \mathcal{X}_2 \longrightarrow 2\mathcal{X}_2$ and $\mathcal{X}_2 \longrightarrow \emptyset$, representing prey reproduction, prey death / predator reproduction, and predator death. The stoichiometry matrix and hazard vector are given by
\[
S=\begin{pmatrix} 1 & -1 & 0\\ 0 & 1 & -1 \end{pmatrix}, \quad 
h(x_t,\theta)=(\theta_1 x_{1,t}, \theta_2 x_{1,t}x_{2,t},\theta_3 x_{2,t})^\top
\]
where $x_{1,t}$ and $x_{2,t}$ denote the respective numbers of prey and predators, at time $t$. Using the MJP representation (\ref{eqn.MJP}), we simulated four data sets using the settings in Table~\ref{tab:tabLVA}. In what follows, we take the tau-leaping model in (\ref{eqn.tauleap}) with $\tau=0.1$ as the inferential model. 

\begin{table}[ht]
\centering
\small
	\begin{tabular}{@{}ll rrr rrr@{}}
         \toprule
Data set & Obs& $\theta$  & $x_0$ &  $\Delta t$ & $T_{max}$ & $T$ & $\frac{1}{T}\sum_{t=1}^{T}p_t$  \\
\midrule
LV20     & $x_{1,t}, x_{2,t}$ & (0.5, 0.0025, 0.3) & (50, 50) & 1\phantom{.5} & 20 & 20 & 0.00136 \\
LV20prey & $x_{1,t}$          & (0.5, 0.0025, 0.3) & (50, 50) & 1\phantom{.5} & 20 & 20 & 0.02836 \\
\midrule
LV40     & $x_{1,t}, x_{2,t}$ & (0.5, 0.0025, 0.3) & (50, 50) & 0.5 & 20 & 40 & 0.00243 \\
LV40prey & $x_{1,t}$          & (0.5, 0.0025, 0.3) & (50, 50) & 0.5 & 20 & 40 & 0.04339 \\
\bottomrule
\end{tabular}
	\caption{Lotka-Volterra tau leap. Synthetic data set name, observed specie(s), ground truth parameter values $\theta$, initial condition $x_0$, inter-observation time $\Delta t$, final time $T_{max}$, number of observations $T$ and average transition probability, $p_t:=\Prob{X_t=x_t|X_{t-\Delta t}=x_{t-\Delta t}}$ (for full observations) or $p_t:=\Prob{Y_t=y_t|Y_{1:t-\Delta t}=y_{1:t-\Delta t}}$ (for partial observations) computed at the ground truth $\theta$.} \label{tab:tabLVA}
\end{table}

For each data set typical transition probabilities are $\mathcal{O}(10^{-3})$ at the ground truth parameter values (and likely to be considerably smaller as $\theta$ deviates from the ground truth). Consequently, implementing the vanilla Frankenfilter (with particle trajectories generated via forward simulation) is computationally prohibitive. Therefore, for each inter-observation interval $(t-\Delta t,t)$, we generate trajectories $x_{(t-\Delta t,t)}^j$ using the conditioned hazard of \cite{GoliWilk15}. For $s\in [t-\Delta t,t)$, the conditioned hazard takes the form
\[
h(x_s,\theta|y_t)= h(x_s,\theta)+H(x_s,\theta)S^\top F\{F^\top SH(x_s,\theta)S^\top F(t-s)\}^{-1}F^\top \{x_t-x_s-Sh(x_s,\theta)(t-s) \}
\]
where $y_t=F^\top x_t$, $H(x_s,\theta)=\text{diag}\{h(x_s,\theta)\}$ and $F=I_{2}$ in the case of full observations, and $F=(1,0)^\top$ in the case of partial observations (of prey values only). Hence, the weight function takes the form
\begin{equation}\label{eqn.weight}
w(x_{(t-\Delta t,t)}^j) = I(F'x_t=F'x_t^j)\prod_{k=1}^{N} \prod_{i=1}^3 
\frac{\mathsf{Po}\left(\Delta n_{i,t-\Delta t+k\tau};\ h(x_{t-\Delta t+(k-1)\tau},\theta)\ \tau\right)}
{\mathsf{Po}\left(\Delta n_{i,t-\Delta t+k\tau};\ h(x_{t-\Delta t+(k-1)\tau},\theta|y_t)\ \tau\right)}
\end{equation}
where $\mathsf{Po}(x; \lambda)$ denotes the Poisson probability mass function with rate $\lambda$, $\Delta n_{i,s}$ denotes the number of reactions of type $i$ in the interval $(s-\tau,s)$ and $N=\Delta t/\tau$.   

In the following numerical experiments, we took an independent prior specification with $\theta_i \sim \mathsf{Gamma}(1,1)$, $i=1,2,3$ and set $m_{-}$ and $m_{+}$ as in Section~\ref{sec:dim}. To validate the tuning advice in Section \ref{sec.tune}, we ran PMMH for $50K$ iterations for each data set for various choices of the required total success $\ess$; Figure~\ref{fig:figLV} reports overall efficiency (ESS/s) against $\ess$. For data sets with full observations, setting $\ess=T$ gives optimal overall efficiency for LV40 and close to optimal overall efficiency for LV20 (of those values of $\ess$ for which overall efficiency is reported). Requiring $V_{rel}=1$ gives $\ess$ values of 31 and 60 for LV20 and LV40 respectively; these correspond to overall efficiencies of approximately $95\%$ of maximum efficiency. For data sets LV20prey and LV40prey, which have observations on prey only, we additionally computed the value of $\ess$ that satisfies \eqref{eqn.Vrel.gen} for $V_{rel}=1$. This was achieved with replicate runs of the Frankenfilter, with $p_t$ approximated via the transition density under the chemical Langevin equation \cite[CLE, see e.g.][]{Wilkinson06}, and the expectations in (\ref{eqn.Vrel.gen}) replaced by their sample equivalents. That is,
\[
p_t \approx \mathsf{N}\left(F^\top x_t; F^\top (x_{t-\Delta t}+\alpha(x_{t-\Delta t},\theta)\Delta t)\,,\, F^\top\beta(x_{t-\Delta t},\theta)F\Delta t \right)
\]
where $\mathsf{N}(\cdot;m,V)$ denotes the Gaussian density with mean vector $m$ and variance matrix $V$, and
\begin{equation}
  \label{eqn.SDEdriftVolFromMJP}
\alpha(x,\theta) = S h(x,\theta), \qquad \beta(x,\theta)=S\text{diag}\{h(x,\theta)\} S^\top.
\end{equation}
Figure~\ref{fig:figLV} (bottom panel) shows the resulting value of $\ess$ (as a vertical line). This value is broadly consistent with the largest overall efficiencies obtained by running the Frankenfilter inside PMMH for various choices of $\ess$, achieving $97\%$ of maximum efficiency for LV20prey and $98\%$ for LV40prey.      

\begin{figure}[ht]
\centering
\includegraphics[width=12.0cm,height=5.0cm]{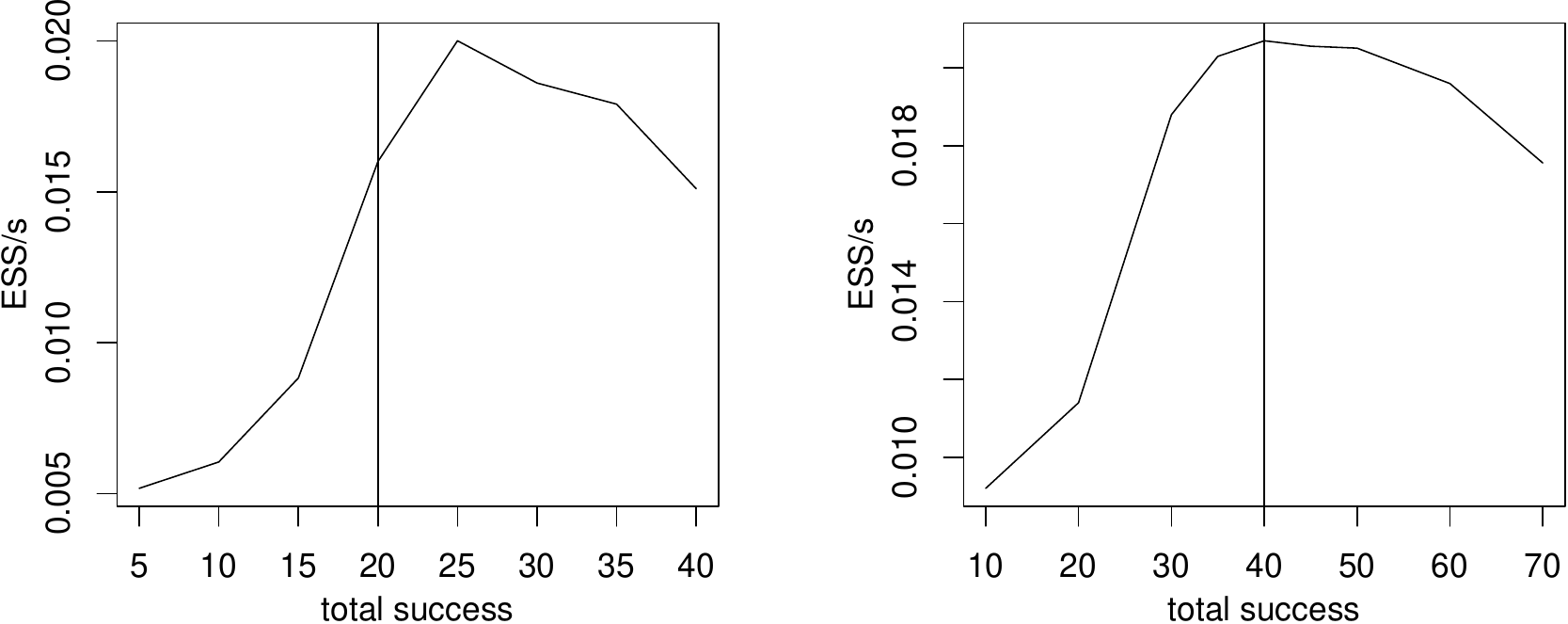}
\\
\includegraphics[width=12.0cm,height=5.0cm]{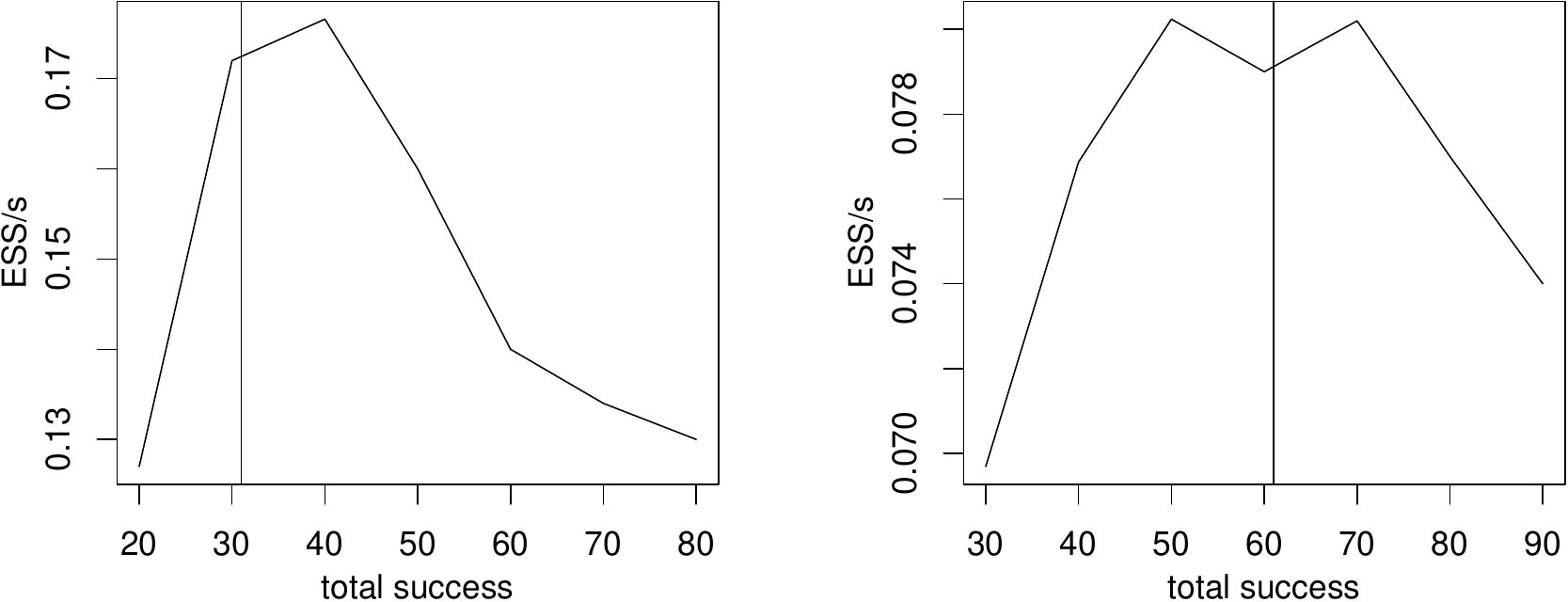}
\caption{Lotka-Volterra tau leap. ESS/s versus required total success $\ess$, from the output of $50K$ iterations of PMMH, using data sets LV20 (top left), LV40 (top right), LV20prey (bottom left) and LV40prey (bottom right). Vertical lines show $\ess=T$ for the full observation data sets LV20 and LV40, and show the  value of $\ess$ obtained by solving \eqref{eqn.Vrel.gen} with $V_{rel}=1$ for the partial observation data sets LV20prey and LV40prey.}
\label{fig:figLV}
\end{figure}

Finally, in Table~\ref{tab:tabLVB}, we compare the overall efficiency of PMMH when using the bootstrap particle filter (BSPF) and Frankenfilter (FF). All filters used the weight function in (\ref{eqn.weight}), the optimal value of $\ess$ (FF) and $m_{+}$ (BSPF), and the PMMH scheme was initialised at an estimate of the posterior mean of $\log\theta$. For the data sets with full observations, FF outperforms BSPF in terms of overall efficiency by a factor of 1.9-2.6. This reduces to 1.7-2.4 for the data sets with partial observations. It is worth noting that these results are predicated on the availability of an initial PMMH run to find key quantities such as the estimated posterior mean and variance, for use in the main monitoring runs. In reality, ``good'' starting values for the PMMH chain may not be readily available. To compare the performance of FF and BSPF over hypothetical initial runs, we ran PMMH using LV20prey, with an initial value given by the ground truth parameter vector multiplied by 5. For FF, we fixed $\ess=2T$ and $m^{+}=3\times 10^5$. For BSPF, we set $m^{+}=3500$, which gave around a $50\%$ chance of generating a non-zero likelihood estimate at the initial value of the chain. Using these settings and $50K$ iterations, overall efficiency scales as $1:4$ for BSPF versus FF. Intuitively, after burn-in, FF is able to use a lower computational effort to meet the desired number of matches.       


\begin{table}[ht]
\centering
\small
	\begin{tabular}{@{}ll rrr rrr@{}}
         \toprule
Data set & Filter  & $\ess$ &  $m_{+}$ & CPU (s) & ESS & ESS/s & Rel \\
\midrule
LV20 & BSPF &   -&  9500          & 244740 & 1893 & 0.0077 & 1.0\\
     & FF   &  25&  $4\times 10^6$& 38163  & 770  & 0.0202 & 2.6\\
\midrule
LV20prey & BSPF                  &   -& 900            & 36532 & 2680 & 0.0733 & 1.0\\
         & FF                    &  40& $3\times 10^5$ & 9840  & 1737 & 0.1766 & 2.4\\
\midrule
LV40 & BSPF &   -&           5400 & 143370 & 1548 & 0.0108 & 1.0\\
     & FF   & 40 & $2\times 10^6$ & 63680  & 1320 & 0.0207 & 1.9\\
\midrule
LV40prey & BSPF                  & -& 850               & 55007 & 2674 & 0.0486 & 1.0\\ 
         & FF                    & 70 & $2\times 10^5$  & 23120 & 1854 & 0.0802 & 1.7\\
\bottomrule
\end{tabular}
	\caption{Lotka-Volterra tau leap. Filter used inside PMMH, required total success $\ess$, maximum number of simulations $m_{+}$, CPU time (in seconds), effective sample size (ESS) and effective sample size per second (ESS/s). All results are based on $50K$ iterations of PMMH.}\label{tab:tabLVB}
\end{table}

\subsection{Lotka--Volterra SDE}
\label{sec.LVSDE}
Here we compare the FF and BSPF for inference on an SDE that is observed with noise. Specifically, we consider the SDE for the Lotka--Volterra model where the drift and volatility are given through \eqref{eqn.SDEdriftVolFromMJP}: $a(X_t)=\alpha(X_t)$ and $b(X_t)=\sqrt{\beta(X_t)}$; for the matrix square root, we used the transpose of the Cholesky decomposition. We employed the same time discretisation as for the tau-leap: $\delta=0.1$.

We created a noisy version of each observation in the LV20 data set : for $i=1,\dots T=20$ and $j=1,2$, $Y_{i,j}|X_{t_i,j}=x\sim \mathsf{Gamma}(\alpha,\alpha/x)$, so that $\Expect{Y_{i,j}|X_{t_i,j}}=X_{t_i,j}$ and $\SD{Y_{i,j}|X_{t_i,j}}=X_{t_i,j}/\sqrt{\alpha}$. We chose $\alpha=50$ to represent relatively precise observations.

Since both the FF and BSPF simulate $X_t$ from its prior, the particle weights are $w_{i}^j=f(y_{i}|x_{t_i}^j)$ where (with subscripts representing vector components here, only)\\ $f(y|x)=\mathsf{Gamma}(y_{1};\alpha,\alpha/x_1)\mathsf{Gamma}(y_{2};\alpha,\alpha/x_2)$. We measure success by
\begin{equation}
  \label{eqn.noisySuccess}
s_{i}^j=\frac{f(y_{i}|x_{t_i}^j)}{\sup_{x\in \cX}f(y_{i}|x)}.
\end{equation}
 One full unit of success is achieved when the state fits the observation as well as it can, and the poorer the fit the lower the amount of success. Appendix \ref{app.general.success} explans why this is a natural measure.

We fixed $\ess=T=20$ and investigated the stability of the variance and the increase in CPU time as the parameter vector fits less and less well with the data.
Specifically, let $\theta_*$ be the true parameter vector, given in Table \ref{tab:tabLVA}. We ran the particle filter $n_{\mathrm{rep}}=500$ times with parameter vector $\theta=\eta \theta_*$, with $\eta$ ranging between $0.72$ and $1.15$. This was performed with $m_+=10^4$ and $m_+=10^3$, denoted as FF4 and FF3, respectively. The BSPF was run with $m=120$ so that at the value of $\eta$ that maximised the likelihood, the sample variance of the log likelihood under the BSPF matched those under the FF. That these variances are $\approx 1$ suggests that setting $\ess=T$ may still be approximately valid for noisy observations when \eqref{eqn.noisySuccess} quantifies success.

\begin{table}
  \small
  \center
  \begin{tabular}{cc|rrr|rr|rr|r}
    \toprule
    $\eta$&$\ell$&$V_{\mathrm{FF4}}$&$V_{\mathrm{FF3}}$&$V_{\mathrm{BS}}$&$\mathrm{CPU}_{\mathrm{FF4}}$&$\mathrm{CPU}_{\mathrm{FF3}}$&$\max\overline{I}_{\mathrm{FF4}}$&$\max\overline{I}_{\mathrm{FF3}}$&$\overline{\overline{I}}_{\mathrm{FF3}}$\\
    \midrule
    0.72&-194.9&9.76&11.3&54.7&721.3&197.9&0.890&1.000&0.186\\
    0.75&-189.3&6.74&8.05&25.0&356.2&146.9&0.482&0.996&0.106\\
    0.80&-183.3&2.72&2.63&6.33&135.4&96.6&0.002&0.854&0.045\\
    0.85&-179.8&1.13&1.14&1.97&78.0&77.4&0&0.094&0.007\\
    0.90&-178.5&0.85&0.89&0.86&73.4&70.3&0&0.356&0.002\\
    0.95&-179.2&1.01&0.95&1.09&87.2&72.2&0&0.758&0.040\\
    1.00&-181.3&1.48&1.72&2.35&114.6&79.5&0.02&0.962&0.048\\
    1.05&-184.7&3.33&3.06&5.15&166.6&98.5&0.010&0.996&0.057\\
    1.10&-189.0&5.76&5.53&13.2&259.5&129.2&0.044&0.998&0.087\\
    1.15&-194.5&9.47&11.3&35.4&537.7&185.6&0.240&1.000&0.127\\
    \bottomrule
  \end{tabular}
  \caption{Lotka-Volterra SDE. The FF used $m_+=10^4$ (FF4) or $m_+=10^3$ (FF3); the BSPF used $m=120$. For each $\eta$, $500$ replicates of each filter were run. The table reports the log-likelihood (to 1dp), the sample variances of the log likelihood estimates for the three filters ($V_{FF4}$, $V_{FF3}$ and $V_{BS}$). The CPU times (in seconds) for the 500 runs of FF4 and FF3, the minimum (over observations) of the fraction of replicates where $\ess$ was not achieved ($\max\overline{I}_{\mathrm{FF4}}$ and $\max\overline{I}_{\mathrm{FF3}}$), and the mean (over observations) of this fraction for FF3 ($\overline{\overline{I}}_{\mathrm{FF3}}$).\label{table.LVvaryMpl}}
\end{table}

For each $\eta$, Table \ref{table.LVvaryMpl} shows the log likelihood, the sample variances of the $n_{rep}$ log likelihood estimates under FF4, FF3 and the BSPF and the CPU time taken by each algorithm; all bootstrap runs took between 59.7 and 61.2 CPU seconds so individual times are not reported separately. For each observation, we also found the fraction of replicates where the target amount of success was not achieved; this provided, for each observation, an estimate of $\Prob{S<\ess}$. Table \ref{table.LVvaryMpl} provides the maximum (over observations)  of these, and for FF3, the average (over observations), which bounds the same under FF4.

Firstly, recall the mathematically optimal sample variances in \cite{SheThiRobRos2015} and \cite{DouPitDelKoh2015} for the case where the idealised Metropolis--Hastings chain would mix slowly: $3.28$ and $2.82$, respectively. However, both analyses assume that $\Var{\ellhat}$ does not vary with $\theta$; since this is not true in practice, and the inefficiency of the underlying chain increases exponentially with the variance, both suggest caution, setting $\Var{\ellhat}\approx 1\text{--}2$ at a representative $\theta$. Table \ref{table.LVvaryMpl} show that this caution is justified: as the log-likelihood decreases to around $16$ below is maximum, $\Var{\ellhat}$ increases by nearly two orders of magnitude for the BSPF , but by only one order of magnitude for the FF. Moreover, it is only when the log likelihood is $10$ below its maximum that FF variance estimates are too large, whereas large variances occur for log-likelihood discrepancies of only $5$  when the BSPF is employed.   

The stability of the FF variance comes at an additional CPU cost whose size depends on $m_+$. The variances with FF3 are similar to those obtained from FF4, yet the more extreme computational costs are avoided, suggesting that $m_+=10^3$ might be preferable and a small but not negligible average failure fraction, rather than a small maximum fraction, might be a better tuning aim.

\subsection{CWD in Mule deer}\label{sec:mule}
Finally, we consider a real data set consisting of the cumulative number
of deaths of mule deer due to chronic wasting disease (CWD) observed at various times
for two separate outbreaks, 1974-1985 and 1992-2001 \cite[see e.g.][]{miller06}. The data have been analysed by \cite{sun15} and \cite{DrovMcC2016} among others; we follow the latter by treating each outbreak as independent and fitting the MJP representation of a susceptible, exposed, infected, removed (SEIR) model (with additional transitions for birth and natural death) to the full data set, denoted by $y=([y_{1:n_1}^{1}]^\top,[y_{1:n_2}^{2}]^\top)^\top$. The associated list of reactions, hazards and stoichiometries is given in Table~\ref{tab:tabSEIRA}. Here, the per capita contact rate is $\beta$, and the average time between exposure and death due to CWD ($1/\mu$) is apportioned between exposure and infection via the parameter $0\leq \alpha \leq 1$. The parameters $a$ and $m$ represent the number of deer added annually and the per-capita per year death date respectively; these are fixed at values included in the data set.    

\begin{table}[ht]
\centering
\small
	\begin{tabular}{@{}llll@{}}
         \toprule
Reaction & Hazard & Stoichiometry & Description \\
\midrule
$\mathcal{R}_1$: $\emptyset \longrightarrow S$ & $a$ & $(1,0,0,0)^\top$ & Birth \\
$\mathcal{R}_2$: $S\longrightarrow \emptyset$ & $m s_t$ & $(-1,0,0,0)^\top$ & Susceptible death (natural) \\
$\mathcal{R}_3$: $S+I\longrightarrow E$ & $\beta s_t i_t$ & $(-1,1,0,0)^\top$ & Contact \\ 
$\mathcal{R}_4$: $E\longrightarrow \emptyset$ & $m e_t$ & $(0,-1,0,0)^\top$ & Exposed death (natural)\\
$\mathcal{R}_5$: $E\longrightarrow I$ & $\frac{\mu}{\alpha} e_t$ & $(0,-1,1,0)^\top$ &  Infection\\
$\mathcal{R}_6$: $I\longrightarrow \emptyset$ & $m i_t$ & $(0,0,-1,0)^\top$ & Infective death (natural)\\
$\mathcal{R}_7$: $I\longrightarrow R$ & $\frac{\mu}{1-\alpha} i_t$ & $(0,0,-1,1)^\top$ & Infective death (CWD)\\
\bottomrule
\end{tabular}
	\caption{SEIR model. Reaction, hazard, stoichiometry (effect on each specie $S_t, I_t, E_t, R_t$) and description.}\label{tab:tabSEIRA}
\end{table}

For the parameters of interest, $\theta=(\beta,\mu,\alpha)^\top$, we followed \cite{DrovMcC2016} by taking an independent prior specification with $\beta\sim\mathsf{Beta}(2,10)$, $\mu\sim\mathsf{Beta}(2,5)$ and $\alpha\sim\mathsf{U}(0,1)$. We fixed $E_1=0$ and specified initial distributions $S_1\sim \mathsf{DiscreteU}(10, 50)$, and $I_1\sim\mathsf{DiscreteU}(0, 20)$ for both outbreaks. 

As with any real data application, the main monitoring run of PMMH with a Gaussian random walk proposal requires a sensible initial parameter value and innovation variance. These quantities can be obtained from the output of a pilot run of PMMH, which in turn requires a sensible $m^+$ for BSPF and additionally $\ess$ for FF. Therefore, using each of 5 values of $\theta$ sampled from the prior, we found values of $m_+$ that gave $\Var{\log \Phat(y)}\approx 1$ and selected the largest value, denoted $m_+^{0}$, for the pilot run. We repeated this process for FF with $m_+$ fixed at $m_+^0$ to find a suitable $\ess$, denoted $\ess^0$. Using these values, pilot runs of PMMH using both BSPF and FF were performed for $50K$ iterations, with the output used to estimate $\Expect{\theta|y}$ and $\Var{\log\theta|y}$. Finally, the main monitoring runs ($50K$ iterations) of PMMH used values $m_+=m_+^1$ and $\ess=\ess^1$ such that $\Var{\log \Phat(y)}\approx 1$ when running BSPF and FF with $\theta$ fixed at $\Expect{\theta|y}$.

\begin{figure}[ht]
\centering
\includegraphics[width=16.0cm,height=5.0cm]{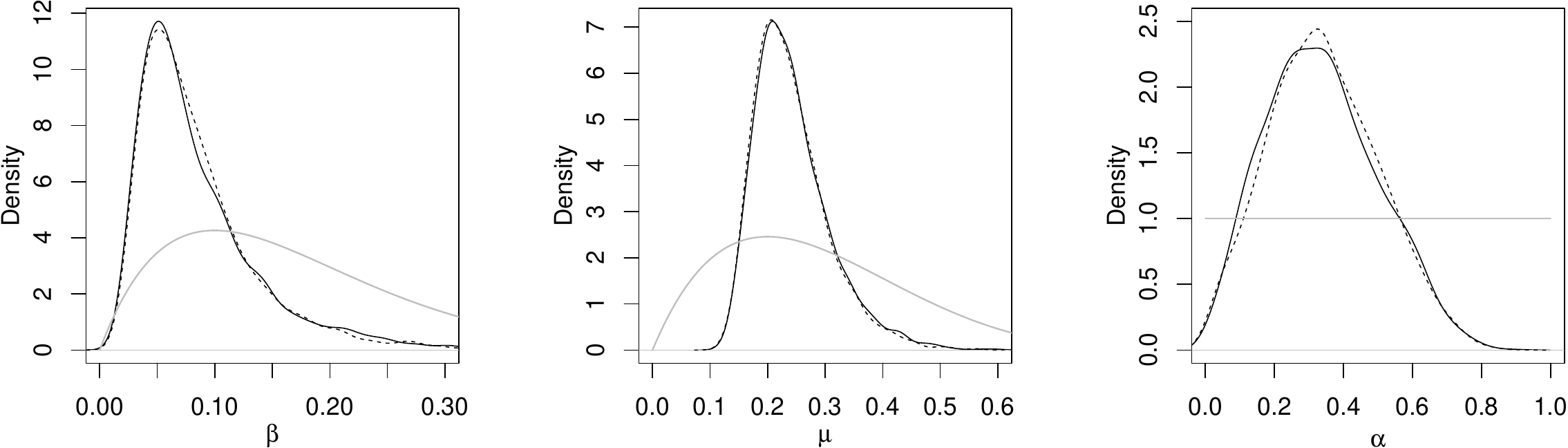}
\caption{SEIR model. Marginal posterior densities of $\beta$, $\mu$ and $\alpha$ based on the output of PMMH with FF (solid black lines) and BSPF (dashed black lines). The prior distribution is shown in grey.}
\label{fig:figSEIR}
\end{figure}

\begin{table}[ht]
\centering
\small
	\begin{tabular}{@{}l rrr rrr rrr@{}}
         \toprule
Filter  & $\ess^0$ &  $m_{+}^0$ & CPU$^0$ (s) & $\ess^1$ & $m_+^1$ & CPU$^1$ (s) & ESS & ESS/TCPU & Rel \\
\midrule
BSPF & - & 6000 & 407628 & - & 1000 & 24777 & 1363 & 0.00315 & 1 \\
FF & 300 & 6000 & \phantom{0}89500 & 90 & 1000 & 15338 & 1328 & 0.01267 & 4\\
\bottomrule
\end{tabular}
	\caption{SEIR model. Filter used inside PMMH, required total success $\ess^0$ ($\ess^1$), maximum number of simulations $m_{+}^0$ ($m_+^1$) and CPU time (in seconds) CPU$^0$ (CPU$^1$) used in the pilot run (main monitoring run), effective sample size (ESS) from the main monitoring run and effective sample size per total CPU time (ESS/TCPU).}\label{tab:tabSEIRB}
\end{table}

Figure~\ref{fig:figSEIR} shows marginal posterior densities for each parameter value, which are consistent across schemes and with the output reported in \cite{DrovMcC2016}. The performance of BSPF and FF can found in Table~\ref{tab:tabSEIRB}. It is evident that, based on the main monitoring run alone, FF outperforms BSPF in terms of overall efficiency by a factor of around 1.6. When the computational cost of the pilot run is taken into account, relative efficiency increases to a factor of 4. That is, during the pilot run under FF and after burn-in, most sampled parameter values require many fewer simulations to obtain the desired number of matches, relative to the number of particles required by BSPF during its pilot run.

\section{Discussion}
\label{sec.Discussion}
We have introduced the Frankenfilter, a close-to-fully alive particle filter that adapts the number of simulations to the difficulty of the problem whilst employing a principled mechanism for bounding the total amount of compute. 

In Section \ref{sec.tune}, we derived advice for choosing the total number of successes, $\ess$, required over each inter-observation interval; in the special case of $T$ exact observations, this simply involves setting $\ess\approx T$. This advice was derived for the alive particle filter, but we also showed that choosing $m_+=\kappa\ess/p$ for a moderately sized $\kappa$ controls the impact of $m_+$ and $m_-$ on the variance of the estimator. Section \ref{sec.sims} validated our the advice on the Frankenfilter.  

The literature backs up the intuition that our advice of choosing $\ess\approx T$ should also apply to PMMH using the alive filter with a large hard threshold. 
\cite{DrovPetMcC2016} uses a hard threshold of $10^5$ to perform inference via particle Metropolis--Hastings. Preliminary runs allow tuning parameters to be set to approximately optimise efficiency. Where the values of $\ess$ are reported, we can see that: in Experiment 1, $\ess=101$ and $T\approx 135$ (before the observations are uninformative zeros), in Experiment 3, $\ess=201$ and $T\approx 240$ and for Experiment 4, $\ess=101$ and $T=100$. Using the same hard threshold, \cite{GoliSher2019} examines a data set with $T=7$ and finds $\ess=8$ to be optimal.

We employed $s(x^j)=f(y|x^j)/\sup_{x\in \cX}f(y|x)$, justified heuristically in Appendix \ref{app.general.success}; however, the Frankenfilter is valid for any choice of $s(\cdot)$, including the constant function, whereupon (as pointed out by a reviewer) it reduces to the standard particle filter. Whilst our choice works well in the simulation studies, it may not be optimal: further work could investigate other possible success functions. Similarly, of our tuning results, the advice derived from Proposition \ref{prop.bound} applies, strictly, to full-vector, exact observations, Theorem \ref{thrm.partialVariance} applies to partial, exact observations, and Propositions \ref{prop.GenOneStepBound} and \ref{prop.BoundProbMissS} require $s(\cdot)\propto w(\cdot)$. Simulation studies suggest that the advice from these, of taking $\ess\approx T$ and choosing $m_+$ such that it is reached on a small but non-negligible fraction of occasions, is at least reasonable, but there is considerable scope for further investigation.  

Comparisons in Section \ref{sec.sims} demonstrated the bias of the alive particle filter with a hard threshold, even when that threshold is high, highlighting the need for the Frankenfilter. The studies tackled a variety of Markov process and filtering scenarios and demonstrated the greater robustness and efficiency of the Frankenfilter when compared with an equivalent standard filter with a fixed number of particles; in some cases the improvement in efficiency was of over an order of magnitude.

In Section \ref{sec.LVSDE}, when $\max \overline{I}_{\mathrm{FF4}}$ is small, the second term in Proposition \ref{prop.GenOneStepBound} is controlled, yet the variance under the Frankenfilter still increases as $\eta$ varies. Whilst the Frankenfilter controls for the difficulty of simulating a good path, it does not control for the positive correlation between the averages of the weights from one observation to the next.

Consistent with the studies of \cite{Alive2015} for the alive particle filter without any threshold, we have provided non-asymptotic bounds on the relative variance of the estimator of the likelihood under certain assumptions, since such quantities are central to the performance of the PMMH algorithm. Central limit theorems (CLTs) are in principle likely to hold under appropriate assumptions, although they are less directly relevant to the use of the algorithm in PMMH. For example, \cite{legland2006sequential} prove a CLT for an estimator of the likelihood that is a biased but an asymptotically negligible perturbation of the estimator considered here, in the case where there is no hard thresholding. It is reasonable to expect that their result can be transferred, when both the hard threshold $m_+$ and $\mathfrak{s}$ diverge to $\infty$. We leave this to future work, as there are several technical details involved in such arguments, as one may deduce from \cite{legland2006sequential}.

Appendix \ref{app.furtherDiscussion} discusses further generalisations and applications of the Frankenfilter.

\textbf{Data Availability Statement}: all data used in this paper were simulated randomly from the models, as specified in Section \ref{sec.sims}.

This research was conducted without external funding.

\textbf{Acknowledgements}: the penultimate simulation study was part-inspired by initial investigations on the use of the Frankenfilter with noisy observations by Andr\'{e} Menezes and the first author.

\bibliographystyle{apalike}

\bibliography{FrankBib.bib}

\begin{thebibliography}{}

\bibitem[Amrein and K\"{u}nsch, 2011]{AmrKun2011}
Amrein, M. and K\"{u}nsch, H.~R. (2011).
\newblock A variant of importance splitting for rare event estimation: Fixed
  number of successes.
\newblock {\em ACM Trans. Model. Comput. Simul.}, 21(2).

\bibitem[Boys et~al., 2008]{BWK08}
Boys, R.~J., Wilkinson, D.~J., and Kirkwood, T. B.~L. (2008).
\newblock Bayesian inference for a discretely observed stochastic kinetic
  model.
\newblock {\em Statistics and Computing}, 18:125--135.

\bibitem[Cao et~al., 2007]{Cao07}
Cao, Y., Gillespie, D.~T., and Petzold, L.~R. (2007).
\newblock Adaptive explicit-implicit tau-leaping method with automatic tau
  selection.
\newblock {\em The Journal of Chemical Physics}, 126(22):224101.

\bibitem[{del Moral} et~al., 2015]{Alive2015}
{del Moral}, P., Jasra, A., Lee, A., Yau, C., and Zhang, X. (2015).
\newblock The alive particle filter and its use in particle {M}arkov chain
  {M}onte {C}arlo.
\newblock {\em Stochastic Analysis and Applications}, 33(6):943--974.

\bibitem[Deligiannidis et~al., 2018]{DDP2018}
Deligiannidis, G., Doucet, A., and Pitt, M.~K. (2018).
\newblock {The Correlated Pseudomarginal Method}.
\newblock {\em Journal of the Royal Statistical Society Series B: Statistical
  Methodology}, 80(5):839--870.

\bibitem[Doucet et~al., 2015]{DouPitDelKoh2015}
Doucet, A., Pitt, M.~K., Deligiannidis, G., and Kohn, R. (2015).
\newblock Efficient implementation of {M}arkov chain {M}onte {C}arlo when using
  an unbiased likelihood estimator.
\newblock {\em Biometrika}, 102(2):295--313.

\bibitem[Drovandi and McCutchan, 2016]{DrovMcC2016}
Drovandi, C.~C. and McCutchan, R.~A. (2016).
\newblock Alive {SMC}2: Bayesian model selection for low-count time series
  models with intractable likelihoods.
\newblock {\em Biometrics}, 72(2):344--353.

\bibitem[Drovandi et~al., 2016]{DrovPetMcC2016}
Drovandi, C.~C., Pettitt, A.~N., and McCutchan, R.~A. (2016).
\newblock {Exact and Approximate Bayesian Inference for Low Integer-Valued Time
  Series Models with Intractable Likelihoods}.
\newblock {\em Bayesian Analysis}, 11(2):325 -- 352.

\bibitem[Gillespie, 1977]{Gillespie77}
Gillespie, D.~T. (1977).
\newblock Exact stochastic simulation of coupled chemical reactions.
\newblock {\em Journal of Physical Chemistry}, 81:2340--2361.

\bibitem[Gillespie, 2001]{Gillespie01}
Gillespie, D.~T. (2001).
\newblock Approximate accelerated stochastic simulation of chemically reacting
  systems.
\newblock {\em Journal of Chemical Physics}, 115(4):1716--1732.

\bibitem[Golightly and Sherlock, 2019]{GoliSher2019}
Golightly, A. and Sherlock, C. (2019).
\newblock Efficient sampling of conditioned {M}arkov jump processes.
\newblock {\em Statistics and Computing}, 29(5):1149--1163.

\bibitem[Golightly and Sherlock, 2022]{GolSher2022}
Golightly, A. and Sherlock, C. (2022).
\newblock Augmented pseudo-marginal {M}etropolis–{H}astings for partially
  observed diffusion processes.
\newblock {\em Statistics and Computing}, 32.

\bibitem[Golightly and Wilkinson, 2015]{GoliWilk15}
Golightly, A. and Wilkinson, D.~J. (2015).
\newblock Bayesian inference for {M}arkov jump processes with informative
  observations.
\newblock {\em SAGMB}, 14(2):169--188.

\bibitem[Gordon et~al., 1993]{GSS1993}
Gordon, N., Salmond, D., and Smith, A. (1993).
\newblock Novel approach to nonlinear/non-{G}aussian {B}ayesian state
  estimation.
\newblock {\em IEE Proceedings F (Radar and Signal Processing)}, 140:107--113.

\bibitem[Kremers, 1987]{Kremers1987}
Kremers, W.~K. (1987).
\newblock An improved estimator of the mean for a sequential binomial sampling
  plan.
\newblock {\em Technometrics}, 29(1):109--112.

\bibitem[Kudlicka et~al., 2020]{KudMurSchLin2020}
Kudlicka, J., Murray, L.~M., Schön, T.~B., and Lindsten, F. (2020).
\newblock Particle filter with rejection control and unbiased estimator of the
  marginal likelihood.
\newblock In {\em ICASSP 2020 - 2020 IEEE International Conference on
  Acoustics, Speech and Signal Processing (ICASSP)}, pages 5860--5864.

\bibitem[Le~Gland and Oudjane, 2006]{legland2006sequential}
Le~Gland, F. and Oudjane, N. (2006).
\newblock A sequential particle algorithm that keeps the particle system alive.
\newblock In {\em Stochastic hybrid systems: Theory and safety critical
  applications}, pages 351--389. Springer.

\bibitem[LeGland and Oudjane, 2004]{LeGlandOudjane2004}
LeGland, F. and Oudjane, N. (2004).
\newblock Stability and uniform approximation of nonlinear filters using the
  {H}ilbert metric and application to particle filters.
\newblock {\em Annals of Applied Probability}, 14(1):144--187.

\bibitem[LeGland and Oudjane, 2005]{LeGlandOudjane2005}
LeGland, F. and Oudjane, N. (2005).
\newblock A sequential particle algorithm that keeps the particle system alive.
\newblock In {\em 2005 13th European Signal Processing Conference}, pages 1--4.

\bibitem[Miller et~al., 2006]{miller06}
Miller, M.~W., Hobbs, N.~T., and Tavener, S.~J. (2006).
\newblock Dynamics of prion disease transmission in mule deer.
\newblock {\em Ecological Applications}, 16:2208--2214.

\bibitem[Pathak, 1976]{Pathak1976}
Pathak, P.~K. (1976).
\newblock Unbiased estimation in fixed cost sequential sampling schemes.
\newblock {\em The Annals of Statistics}, 4(5):1012--1017.

\bibitem[Pitt et~al., 2012]{PitSilGioKoh2012}
Pitt, M.~K., dos Santos~Silva, R., Giordani, P., and Kohn, R. (2012).
\newblock On some properties of {M}arkov chain {M}onte {C}arlo simulation
  methods based on the particle filter.
\newblock {\em Journal of Econometrics}, 171(2):134--151.
\newblock Bayesian Models, Methods and Applications.

\bibitem[Pitt and Shephard, 1999]{PittShep1999}
Pitt, M.~K. and Shephard, N. (1999).
\newblock Filtering via simulation: Auxiliary particle filters.
\newblock {\em Journal of the American Statistical Association},
  94(446):590--599.

\bibitem[Schmon et~al., 2021]{schmon2021}
Schmon, S.~M., Deligiannidis, G., Doucet, A., and Pitt, M.~K. (2021).
\newblock Large sample asymptotics of the pseudo-marginal method.
\newblock {\em Biometrika}, 108(1):37--51.

\bibitem[Shelley, 1818]{Shelley1818}
Shelley, M. (1818).
\newblock {\em Frankenstein; or, The Modern Prometheus}.
\newblock Lackington, Hughes, Harding, Mavor and Jones, Finsbury Square,
  London, UK.

\bibitem[Sherlock, 2024]{SheVB2024}
Sherlock, C. (2024).
\newblock Variance bounds and robust tuning for pseudo-marginal
  {M}etropolis--hastings algorithms.

\bibitem[Sherlock et~al., 2015]{SheThiRobRos2015}
Sherlock, C., Thiery, A.~H., Roberts, G.~O., and Rosenthal, J.~S. (2015).
\newblock On the efficiency of pseudo-marginal random walk {M}etropolis
  algorithms.
\newblock {\em Annals of Statistics}, 43(1):238--275.

\bibitem[Sisson et~al., 2018]{Sisson2018ABC}
Sisson, S.~A., Fan, Y., and Beaumont, M.~A., editors (2018).
\newblock {\em Handbook of Approximate Bayesian Computation}.
\newblock Chapman \& Hall/CRC, New York.

\bibitem[Sun et~al., 2015]{sun15}
Sun, L., Lee, C., and Hoeting, J.~A. (2015).
\newblock Parameter inference and model selection in deterministic and
  stochastic dynamical models via approximate {B}ayesian computation: modeling
  a wildlife epidemic.
\newblock {\em Environmetrics}, 26:451--462.

\bibitem[Vats et~al., 2019]{vats19}
Vats, D., Flegal, J.~M., and Jones, G.~L. (2019).
\newblock Multivariate output analysis for {M}arkov chain {M}onte {C}arlo.
\newblock {\em Biometrika}, 106(2):321--337.

\bibitem[Wilkinson, 2018]{Wilkinson06}
Wilkinson, D.~J. (2018).
\newblock {\em Stochastic Modelling for Systems Biology}.
\newblock Chapman \& Hall/CRC Press, Boca Raton, Florida, 3rd edition.

\end{thebibliography}

\appendix

\section{Proof of Theorem \ref{thrm.unbiased}}
\label{sec.proof.unbiased.gen}
Our proof is inspired by \cite{PitSilGioKoh2012} and \cite{KudMurSchLin2020}.

Let $\cF_{t}$ comprise all of the random variables simulated up to and including time $t$. We define $\Phat(y_1)$ to be the random variable whose realisation is $\phat_1$. For $t\ge 2$, we define $\Phat(y_{t-h:t}|y_{1:t-h-1})$ to be the random variable whose realisation is $\prod_{j=t-h}^t \phat_j$, with $\Phat(y_t|y_{1:t-1}):=\Phat(y_{t:t}|y_{1:t-1})$. We also define $I_t=1$ if $K_t=1$ and $I_t=0$ if $K_t\in\{0,2\}$. 

Analogously to the single-observation case, if $I_t=1$ then $(W_t^1,S_t^1),\dots,(W_t^{M_t-1},S_t^{M_t-1})$ are exchangeable, and if $I_t=0$, $(W_t^1,S_t^1),\dots,(W_t^{M_t},S_t^{M_t})$ are exchangeable. 

\begin{lemma}
  \label{lemma.base}
  \[
  \Expect{\Phat(y_t|y_{1:t-1})|\cF_{t-1}}=
i_{t-1}\frac{\sum_{j=1}^{m_{t-1}-1}w_{t-1}^j\Prob{Y_t=y_t|x_{t-1}^j}}{\sum_{k=1}^{m_{t-1}-1}w_{t-1}^k}
+
(1-i_{t-1})\frac{\sum_{j=1}^{m_{t-1}}w_{t-1}^j\Prob{Y_t=y_t|x_{t-1}^j}}{\sum_{k=1}^{m_{t-1}}w_{t-1}^k}.
  \]
\end{lemma}

\textbf{Proof}: By exchangeability,
\begin{align*}
  \Expect{\Phat(y_t|y_{1:t-1})|\cF_{t-1},I_t=1}
  &=
  \Expect{\frac{\sum_{j=1}^{M_t-1}W_{t}^j}{M_t-1}|\cF_{t-1},I_t=1}
  =
  \Expect{W_t^1|\cF_{t-1},I_t=1}.\\
  \Expect{\Phat(y_t|y_{1:t-1})|\cF_{t-1},I_t=0}
  &=
\Expect{\frac{\sum_{j=1}^{M_t}W_{t}^j}{M_t}|\cF_{t-1},I_t=0}
=
\Expect{W_t^1|\cF_{t-1},I_t=0}.
\end{align*}
Therefore
\begin{align*}
\Expect{\Phat(y_t|y_{1:t-1})|\cF_{t-1}}
&=
\Expect{\Expect{\Phat(y_t|y_{1:t-1})|I_t}|\cF_{t-1}}
=
\Expect{\Expect{W_t^1|I_t,\cF_{t-1}}|\cF_{t-1}}
=
\Expect{W^1_t|\cF_{t-1}}.
\end{align*}
But $\Expect{W_t^i|\cF_{t-1},a_{t-1}^i}=\Prob{Y_t=y_t|x_{t-1}^{a_{t-1}^i}}$, so
\begin{align*}
  \Expect{W^1_t|\cF_{t-1}}
&=
\Expect{\Expect{W^1_t|A_{t-1}^1}|\cF_{t-1}}
=
\Expect{\Prob{Y_t=y_t|x_{t-1}^{A_{t-1}^1}}|\cF_{t-1}}\\
&=
i_{t-1}\frac{\sum_{j=1}^{m_{t-1}-1}w_{t-1}^j\Prob{Y_t=y_t|x_{t-1}^j}}{\sum_{k=1}^{m_{t-1}-1}w_{t-1}^k}
+
(1-i_{t-1})\frac{\sum_{j=1}^{m_{t-1}}w_{t-1}^j\Prob{Y_t=y_t|x_{t-1}^j}}{\sum_{k=1}^{m_{t-1}}w_{t-1}^k}.~\square
\end{align*}

\begin{lemma}
  \label{lemma.induction}
  \begin{align*}
  \Expect{\Phat(y_{t-h:t}|y_{1:t-h-1})|\cF_{t-h-1}}&=
i_{t-h-1}\frac{\sum_{j=1}^{m_{t-h-1}-1}w_{t-h-1}^j\Prob{Y_{t-h:t}=y_{t-h:t}|x_{t-h-1}^j}}{\sum_{k=1}^{m_{t-h-1}-1}w_{t-h-1}^k}\\
&+
(1-i_{t-h-1})\frac{\sum_{j=1}^{m_{t-h-1}}w_{t-h-1}^j\Prob{Y_{t-h:t}=y_{t-h:t}|x_{t-h-1}^j}}{\sum_{k=1}^{m_{t-h-1}}w_{t-h-1}^k}.
\end{align*}
\end{lemma}

\textbf{Proof}: The proof is by induction. The case of $h=0$ is covered by Lemma \ref{lemma.base}. We now assume that it holds for $h$ and show that it holds for $h+1$.

Replacing $h$ with $h+1$ in the left hand side of the equality in the lemma statement,
\begin{align*}
  \Expect{\Phat(y_{t-h-1:t}|y_{1:t-h-2})|\cF_{t-h-2}}
  &=
  \Expect{\Phat(y_{t-h:t}|y_{1:t-h-1})\Phat(y_{t-h-1}|y_{1:t-h-2})|\cF_{t-h-2}}\\
  &=
  \Expect{\Expect{\Phat(y_{t-h:t}|y_{1:t-h-1})\Phat(y_{t-h-1}|y_{1:t-h-2})|\cF_{t-h-1}}\cF_{t-h-2}}\\
  &=
  \Expect{\Phat(y_{t-h-1}|y_{1:t-h-2})\Expect{\Phat(y_{t-h:t}|y_{1:t-h-1})|\cF_{t-h-1}}\cF_{t-h-2}}.
\end{align*}
Now, if $I_{t-h-1}=1$, from the definition of $\Phat(y_{t-h-1}|y_{1:t-h-2})$ and the inductive hypothesis,
\begin{align*}
\Phat(y_{t-h-1}|y_{1:t-h-2})
  &=
  \frac{1}{M_{t-h-1}-1}\sum_{k=1}^{M_{t-h-1}-1}W^k_{t-h-1}
  \\
  \Expect{\Phat(y_{t-h:t}|y_{1:t-h-1})|\cF_{t-h-1}}
  &=
  \frac{\sum_{j=1}^{M_{t-h-1}-1}W_{t-h-1}^j\Prob{Y_{t-h-1:t}=y_{t-h:t}|X_{t-h-1}^j}}{\sum_{k=1}^{M_{t-h-1}-1}W_{t-h-1}^k}.
\end{align*}
Thus,
\begin{align*}
  \Phat(y_{t-h-1}|y_{1:t-h-2})    \Expect{\Phat(y_{t-h:t}|y_{1:t-h-1})|\cF_{t-h-1}}&=
   \frac{1}{M_{t-h-1}-1}\sum_{j=1}^{M_{t-h-1}-1}W_{t-h-1}^j\Prob{Y_{t-h:t}=y_{t-h:t}|X_{t-h-1}^j}.
\end{align*}
Exchangeability then gives
\begin{align*}
  \Expect{\Phat(y_{t-h-1:t}|y_{1:t-h-2})|\cF_{t-h-2},I_{t-h-1}=1}
  &=\Expect{W_{t-h-1}^1\Prob{Y_{t-h:t}=y_{t-h:t}|X_{t-h-1}^1}|\cF_{t-h-2},I_{t-h-1}=1}.
\end{align*}
When $I_{t-h-1}=0$, $M_{t-h-1}-1$ is replaced by $M_{t-h-1}$ in the above, and analogous steps give
\begin{align*}
  \Expect{\Phat(y_{t-h-1:t}|y_{1:t-h-2})|\cF_{t-h-2},I_{t-h-1}=0}
  =\Expect{W_{t-h-1}^1\Prob{Y_{t-h:t}=y_{t-h:t}|X_{t-h-1}^1}|\cF_{t-h-2},I_{t-h-1}=0}.
\end{align*}
So
\begin{align*}
\Expect{\Phat(y_{t-h-1:t}|y_{1:t-h-2})|\cF_{t-h-2}}
&=
\Expect{\Expect{\Phat(y_{t-h-1:t}|y_{1:t-h-2})|\cF_{t-h-2},I_{t-h-1}}|\cF_{t-h-2}}\\
&=
\Expect{\Expect{W_{t-h-1}^1\Prob{Y_{t-h:t}=y_{t-h:t}|X_{t-h-1}^1}|\cF_{t-h-2},I_{t-h-1}}|\cF_{t-h-2}}\\
&=
\Expect{W_{t-h-1}^1\Prob{Y_{t-h:t}=y_{t-h:t}|X_{t-h-1}^1}|\cF_{t-h-2}}.
\end{align*}
Conditional on $\cF_{t-h-2}$,
\[
W_{t-h-1}^1= \frac{p(X^1_{(t-h-2,t-h-1]}|x_{t-h-2}^{A^1_{t-h-2}})}{q(X^1_{(t-h-2,t-h-1]}|x_{t-h-2}^{A^1_{t-h-2}})}
 \Prob{Y_{t-h-1}=y_{t-h-1}|X_{t-h-1}^1},
 \]
 so $\Expect{W_{t-h-1}^1\Prob{Y_{t-h:t}=y_{t-h:t}|X_{t-h-1}^1}|\cF_{t-h-2}}$ is
\begin{align*}
\Expect{
 \frac{p(X^1_{(t-h-2,t-h-1]}|x_{t-h-2}^{A^1_{t-h-2}})}{q(X^1_{(t-h-2,t-h-1]}|x_{t-h-2}^{A^1_{t-h-2}})}
 \Prob{Y_{t-h-1:t}=y_{t-h-1:t}|X_{t-h-1}^1}
  |\cF_{t-h-2}}.
\end{align*}
Finally, $i_{t-h-2}=1$, $X^1_{(t-h-2,t-h-1]}$ is sampled by choosing an ancestor $j:=a_{t-h-2}^1$ with probability proportional to $w^j_{t-h-2}$ with $j\in \{1,\dots,m_{t-h-2}-1\}$ if $i_{t-h-2}=1$ and $j\in \{1,\dots,m_{t-h-2}\}$ if $i_{t-h-2}=0$, then propagating according to $q(x_{(t-h-2,t-h-1]}|x_{t-h-2}^j)$. Thus, if $i_{t-h-2}=1$, the required expectation is
\[
\frac{1}{\sum_{k=1}^{m_{t-h-2}-1}w_{t-h-2}^k}
\sum_{j=1}^{m_{t-h-2}-1}w_{t-h-2}^j~\Expect{ \frac{p(X_{(t-h-2,t-h-1]}|x_{t-h-2}^j)}{q(X_{(t-h-2,t-h-1]}|x_{t-h-2}^j)}
 \Prob{Y_{t-h-1:t}=y_{t-h-1:t}|X_{t-h-1}}
}
\]
which is
\[
\frac{1}{\sum_{k=1}^{m_{t-h-2}-1}w_{t-h-2}^k}
\sum_{j=1}^{m_{t-h-2}-1}w_{t-h-2}^j ~\Prob{Y_{t-h-1:t}=y_{t-h-1:t}|x_{t-h-2}^j}.
\]
Analogously, if $i_{t-h-2}=0$, we obtain
\[
\frac{1}{\sum_{k=1}^{m_{t-h-2}}w_{t-h-2}^k}
\sum_{j=1}^{m_{t-h-2}}w_{t-h-2}^j ~\Prob{Y_{t-h-1:t}=y_{t-h-1:t}|x_{t-h-2}^j}.
\]
These give the formula stated in the lemma but with $h$ replaced by $h+1$, as required. $\square$

Finally, we are in a position to prove Theorem \ref{thrm.unbiased} itself.

Substitute $h=t-2$ into Lemma \ref{lemma.induction}, to give
  \begin{align*}
  \Expect{\Phat(y_{2:t}|y_1)|\cF_{1}}&=
i_{1}\frac{\sum_{j=1}^{m_{1}-1}w_{1}^j\Prob{Y_{2:t}=y_{2:t}|x_{1}^j}}{\sum_{k=1}^{m_{1}-1}w_{1}^k}
+
(1-i_{1})\frac{\sum_{j=1}^{m_{1}}w_{1}^j\Prob{Y_{2:t}=y_{2:t}|x_{1}^j}}{\sum_{k=1}^{m_{1}}w_{1}^k}.
\end{align*}
  Given the form of $\Phat(y_1)$ (conditional on $i_1$) and the exchangeability of $W_1^j$ for $j \in \{1,\dots,M_1-1\}$ if $i_1=1$ and $j \in \{1,\dots,M_1\}$ if $i_1=0$, we can take expectations over $\cF_1$ to obtain
\begin{align*}
  \Expect{\Phat(y_{1:t})}&=\Expect{\Phat(y_1)\Phat(y_{2:t}|y_1)}
=
  \Expect{w_1(X_1^1)\Prob{Y_{2:t}=y_{2:t}|X_1^1}}\\
  &=
  \Expects{X_0\sim p_0,X_{(0,1]}\sim q(\cdot|X_0)}{\frac{p(X_{(0,1]})}{q(X_{(0,1]})}\Prob{Y_1=y_1|X_1}\Prob{Y_{2:t}=y_{2:t}|X_1}}\\
  &=
  \Expects{X_0\sim p_0,X_{(0,1]}\sim q(\cdot|X_0)}{\frac{p(X_{(0,1]})}{q(X_{(0,1]})}\Prob{Y_{1:t}=y_{1:t}|X_1}}
  =
  \Prob{Y_{1:t}=y_{1:t}}.~\square
  \end{align*}

\section{Proofs of Propositions \ref{prop.bound}, \ref{prop.GenOneStepBound} and \ref{prop.BoundProbMissS}}
\label{sec.proof.PropBound}

\subsection{Proof of Proposition \ref{prop.bound}}
Firstly, let $X=M_t-1$. We are interested in $\frac{1}{p^2}\Expect{\{(s-1)/X\}^2}=\frac{(s-1)^2}{p^2}\Expect{1/X^2}$. 

The random variable $X$ is a shifted version of a negative binomial random variable, and so it has a probability mass function of
\begin{equation}
  \label{eqn.pmfMm1}
\Prob{X=x}= \frac{x!}{(\ess-1)!(x+1-s)!}p^{\ess}(1-p)^{x+1-\ess},~~~(x=\ess-1,\ess,\ess+1,\dots),
\end{equation}
and $\Prob{X=x}=0$ otherwise. Thus
\[
\Expect{\frac{1}{X^2}}
=
\frac{p^{\ess}}{(1-p)^{\ess-1}(\ess-1)!}\sum_{x=\ess-1}^\infty\frac{(x-1)!}{x(x+1-\ess)!}(1-p)^x.
\]

\textbf{First case} ($\ess=2$): here
\[
\Expect{\frac{1}{X^2}}
=
\frac{p^2}{1-p}
\sum_{x=1}^\infty \frac{1}{x}(1-p)^x
=
-\frac{p^2}{1-p}\log(p),
\]
by matching the infinite sum to the Taylor expansion of $-\log[1-(1-p)]$.  Multipying by $(\ess-1)^2/p^2$ gives the required result.

\textbf{Second case} ($\ess=3$): here
\[
\Expect{\frac{1}{X^2}}
=
\frac{p^3}{2(1-p)^2}
A,
\]
where
\[
A=\sum_{x=2}^\infty \frac{x-1}{x}(1-p)^x
=
\sum_{x=2}^\infty (1-p)^x
-
\sum_{x=2}^\infty \frac{1}{x}(1-p)^x
=
\frac{(1-p)^2}{p}- \{-\log p - (1-p)\},
\]
where, for the second sum, we have matched the same Taylor expansion as when $\ess=2$. Multiplying by $(\ess-1)^2/p^2 \times p^3/\{2(1-p)^2\}=2p/(1-p)^2$ gives the exact sum.

To obtain the bounds, let 
\[
S:=\sum_{x=2}^\infty \frac{1}{x}(1-p)^x.
\]
Then, for $0<p\le 1$,
\[
\frac{(1-p)^2}{2}\le S\le \frac{(1-p)^2}{2}+\frac{(1-p)^3}{3}\left[\sum_{j=0}^\infty (1-p)^j\right]
=
\frac{(1-p)^2}{2}+\frac{(1-p)^3}{3p}.
\]
Thus
\[
\frac{(1-p)^2}{p}-\frac{(1-p)^2}{2}-\frac{(1-p)^3}{3p}
\le A
\le
\frac{(1-p)^2}{p}-\frac{(1-p)^2}{2}.
\]
Multiplying by $2p/(1-p)^2$ gives the bounds.

\textbf{General case}. Now,
\[
\frac{1}{X(X-1)}-\frac{a}{X(X-1)(X-2)}-\frac{1}{X^2}=\frac{(1-a)X-2}{X^2(X-1)(X-2)}.
\]
Since $X\ge s-1$, setting $a=1$ and then $a=1-2/(s-1)$,
\begin{equation}
  \label{eqn.base.ineq}
\frac{1}{X(X-1)}-\frac{1}{X(X-1)(X-2)}
<\frac{1}{X^2}
  \le\frac{1}{X(X-1)}-\frac{1-\frac{2}{s-1}}{X(X-1)(X-2)}.
  \end{equation}
 In particular, therefore, as well as $\Expect{1/X}=p/(s-1)$, which must be true for unbiasedness of the Alive Particle Filter, for $s\ge 3$,
\[
\Expect{\frac{1}{X(X-1)}}=\frac{p^2}{(s-1)(s-2)}
=
\frac{p^2}{(s-1)^2}\left(1+\frac{1}{s-2}\right).
\]
Also, for $s\ge 4$,
\begin{align*}
\Expect{\frac{1}{X(X-1)(X-2)}}&=\frac{p^3}{(s-1)(s-2)(s-3)}
=
\frac{p^2}{(s-1)^2}\times\frac{p(s-1)}{(s-2)(s-3)}\\
&=
\frac{p^2}{(s-1)^2}\times\frac{p}{s-2}\left(1+\frac{2}{s-3}\right).
\end{align*}
Thus, taking expectations in \eqref{eqn.base.ineq} gives,
\[
1+\frac{1-p(1+\frac{2}{s-3})}{s-2}
<
\frac{(s-1)^2}{p^2}\Expect{\frac{1}{X^2}}
<
1+\frac{1}{s-2}-\frac{p}{s-2}\left(1-\frac{2}{s-1}\right)\left(1+\frac{2}{s-3}\right),
\]
which simplifies to the required expression. $\square$

\subsection{Proof of Proposition \ref{prop.GenOneStepBound}}
For a general $c$, taking $S^j\gets cS^j$ and $\ess \gets c\ess$ reduces to the case $c=1$. Thus, we set $c=1$; reversing the above map for the final results.

Let $K\in\{0,1,2\}$ be as defined in the proof of Proposition \ref{prop.GenOneUnbiased} and let $A_K$, $K\in\{0,1,2\}$ be defined below:
$$
\Expect{P^2}=\underbrace{\Expect{P^2\ind{K=0}}}_{A_0}+\underbrace{\Expect{P^2\ind{K=1}}}_{A_1}+\underbrace{\Expect{P^2\ind{K=2}}}_{A_2}.
$$
Firstly, if $K=2$ then $\sum_{j=1}^{m_+} W^j<\ess$, so $P<\ess/m_+$. Also, $\Expect{P|K=2}=\Expect{P|P<\ess/m_+}\le \Expect{P}=p$ (see Proposition \ref{prop.GenOneUnbiased}). Thus,
$$
A_2\le \frac{\ess}{m_+}\Expect{P\ind{K=2}}
=
\frac{\ess}{m_+}\Expect{P|K=2}\Prob{K=2}
\le
\frac{\ess p}{m_+}\Prob{K=2}.
$$
Next, using the exchangeability of $W^1,\dots W^{M-1}$ when $K=1$,
$$
\begin{aligned}
A_1&=
\Expect{\left(\frac{1}{M-1}\sum_{j=1}^{M-1}W^j\right)^2\ind{K=1}}\\
&=
\Expect{\frac{1}{M-1}\sum_{j=1}^{M-1}W^j\left(\frac{1}{M-1}\sum_{j=1}^{M-1}W^j\right)\ind{K=1}}\\
&=
\Expect{W^1\frac{1}{M-1}\sum_{j=1}^{M-1}W^j~~\ind{K=1}}.
\end{aligned}
$$
Now, $W^1 \le w_*$ and when $K=1$, $\sum_{j=2}^{M-1}W^j \ge \ess-2w_*$, so
$$
\sum_{j=1}^{M-1}W^j=\sum_{j=2}^{M-1}W^j\times \left(1+\frac{W^1}{\sum_{j=2}^{M-1}W^j}\right)
\le
\sum_{j=2}^{M-1}W^j \times \left(1+\frac{w_*}{\ess-2w_*}\right).
$$
Thus
$$
\begin{aligned}
A_1 &\le \frac{\ess-w_*}{\ess-2w_*}\Expect{\frac{M-2}{M-1}~W^1~ \frac{1}{M-2}\sum_{j=2}^{M-1}W^j~\ind{K=1}}
\le
\frac{\ess-w_*}{\ess-2w_*}\Expect{W^1~ \frac{1}{M-2}\sum_{j=2}^{M-1}W^j~~\ind{K=1}}\\
&=
\frac{\ess-w_*}{\ess-2w_*}\Expect{W^1 W^2 ~~\ind{K=1}},
\end{aligned}
$$
by the exchangeability of the $W^j$.

Similarly,
$$
A_0=\Expect{\left(\frac{1}{M}\sum_{j=1}^{M}W^j\right)^2\ind{K=0}}
=
\Expect{W^1\frac{1}{M}\sum_{j=1}^{M}W^j~~\ind{K=0}}
$$
and
$$
\sum_{j=1}^{M}W^j
\le
\sum_{j=2}^{M}W^j \times \left(1+\frac{w_*}{\ess-w_*}\right).
$$
Thus
$$
A_0\le \frac{\ess}{\ess-w_*}\Expect{W^1 W^2 ~~\ind{K=0}}\le \frac{\ess-w_*}{\ess-2w_*}\Expect{W^1 W^2 ~~\ind{K=0}}.
$$

Combining the bounds for $A_0$, $A_1$ and $A_2$,
$$
\Expect{P^2}
\le
\frac{\ess-w_*}{\ess-2w_*}\Expect{W^1W^2\ind{K\ne 2}}
+
\frac{\ess p}{m_+}\Prob{K=2}.
$$
However, $\Expect{W^1W^2\ind{K\ne 2}}\le \Expect{W^1W^2}=\Expect{W^1}\Expect{W^2}$, by independence. Hence
$$
\Expect{P^2}
\le
\left(1+\frac{w_*}{\ess-2w_*}\right) p^2
+
\frac{\ess p}{m_+}\Prob{K=2},
$$
as required. 

\subsection{Proof of Proposition \ref{prop.BoundProbMissS}}
Firstly, $\Prob{\sum_{j=1}^{m_+}S^j<\ess}=\Prob{\sum_{j=1}^{m_+}W^j<c\ess}$. 

Next, $\Prob{\sum_{j=1}^{m_+}W^j -m_+ p< -t}\le  \Prob{\sum_{j=1}^{m_+}W^j -m_+ p\le -t}$. 
Now, $|W^j-\Expect{W^j}|<w_*$, so, whenever $m_+p>c\ess$, Bernstein's inequality provides
$$
\Prob{\sum_{j=1}^{m_+}W^j -m_+ p\le -t}\le \exp\left\{-\frac{\frac{1}{2}t^2}{m_+ \Var{W^1}+\frac{1}{3}w_*t}\right\}.
$$
Setting $t=m_+ p - c\ess$, we have
\[
\Prob{\sum_{j=1}^{m_+}W^j\le c\ess}\le \exp\left\{-\frac{\frac{1}{2}(m_+p-c\ess)^2}{m_+ \Var{W^1}+\frac{1}{3}w_*(m_+p-c\ess)}\right\}.
\]
Now $\Var{W^1}\le \Expect{(W^1)^2}\le w_*\Expect{W^1}=w_*p$, so,
\[
\Prob{\sum_{j=1}^{m_+}W^j\le c\ess}\le \exp\left\{-\frac{1}{w_*}~\frac{\frac{1}{2}(m_+p-c\ess)^2}{m_+ p+\frac{1}{3}(m_+p-c\ess)}\right\}.
\]
Setting $m_+p=\kappa c\ess$ for some $\kappa >1$ gives
$$
\Prob{\sum_{j=1}^{m_+}W^j\le c\ess}\le \exp\left\{-\frac{1}{w_*}~\frac{\frac{3}{2}(\kappa-1)^2 c\ess}{4\kappa -1}\right\}.
$$
For $\kappa\ge 7/4$, $(\kappa-1)^2-(4\kappa-1)(\kappa-7/4)/4=9/16>0$, so $(\kappa-1)^2/(4\kappa-1)> (\kappa-7/4)/4$, from which
$$
\Prob{\sum_{j=1}^{m_+}W^j\le c\ess}<
\exp\left\{-\frac{3(\kappa-\frac{7}{4}) c\ess}{8w_*}\right\}.
$$

\section{Proof of Theorem \ref{thrm.partialVariance}}
\label{sec.prove.partialVariance}
The relative variance is $E_2/p^2 -1$, where $E_2:=\Expect{\prod_{t=1}^T\Phat_t^2}$. 

We show that 
\begin{equation}
\label{eqn.inductive}
\Expect{\left(\prod_{s=1}^t\Phat_s^2\right)\left(\prod_{s=t+1}^T\Pbar_s^2\right)}
\le 
\left(1+\frac{1}{\ess-2}\right)
\Expect{\left(\prod_{s=1}^{t-1}\Phat_s^2\right)\left(\prod_{s=t}^T\Pbar_s^2\right)}.
\end{equation}
Applying this recursively from $t=T$ to $t=1$ gives
\[
E_2\le \left(1+\frac{1}{\ess-2}\right)^{T}\Expect{\prod_{t=1}^T\Pbar_t^2}.
\]
The result then follows since $1+\frac{1}{\ess-2}<\exp\{1/(\ess-2)\}$.

The left hand side (LHS) of \eqref{eqn.inductive} is
$$
\begin{aligned}
    \mathrm{LHS}&
    =
\Expect{\Expect{
\left(\prod_{s=1}^t\Phat_s^2\right)\left(\prod_{s=t+1}^T\Pbar_s^2\right)|\cF_{t-1}}}\\
&=
\Expect{\left(\prod_{s=1}^{t-1}\Phat_s^2\right)\Expect{\Phat_t^2\left(\prod_{s=t+1}^T\Pbar_s^2\right)|\cF_{t-1}}}\\
&=
\Expect{\left(\prod_{s=1}^{t-1}\Phat_s^2\right)\Expect{\Phat_t^2|\cF_{t-1}}\Expect{\left(\prod_{s=t+1}^T\Pbar_s^2\right)|\cF_{t-1}}}
~~~(\mbox{by S2})\\
&\le
\left(1+\frac{1}{\ess-2}\right)
\Expect{\left(\prod_{s=1}^{t-1}\Phat_s^2\right)\Pbar_t^2\Expect{\left(\prod_{s=t+1}^T\Pbar_s^2\right)|\cF_{t-1}}}~~~(\mbox{by \eqref{eqn.usefulInterim}})\\
&=
\left(1+\frac{1}{\ess-2}\right)
\Expect{\Expect{\left(\prod_{s=1}^{t-1}\Phat_s^2\right)\Pbar_t^2\left(\prod_{s=t+1}^T\Pbar_s^2\right)|\cF_{t-1}}}\\
&=
\left(1+\frac{1}{\ess-2}\right)
\Expect{\left(\prod_{s=1}^{t-1}\Phat_s^2\right)\Pbar_t^2\left(\prod_{s=t+1}^T\Pbar_s^2\right)}.~~~\square
\end{aligned}
$$

\section{On general measures of success}
\label{app.general.success}
When we have exact, full-vector or exact, partial observations the likelihood is $f(y_{i}|x_{t_i}^j)=1(h(x_{t_i}^j)=y_i)$, for some function $h$ using a subset or all of the components of $x$.

For the noisily observed process in Section \ref{sec.LVSDE} we used
\[
s_i^j=\frac{f(y_{i}|x_{t_i}^j)}{\sup_{x\in \cX}f(y_i|x)}.
\]
To explain why this is natural, it is helpful to define $w_i^j:=f(y_{i}|x_{t_i}^j)$ and $w_i^*:=\sup_{x\in \cX}f(y_i|x)$.

An alternative to simply defining $w_i^j$ as above, would be to use a random weight, $\Wtil_i^j$, where $\Wtil_i^j=w_i^*$ with probability $w_i^j/w_i^*$ and $\Wtil_i^j=0$ otherwise. Since $\Expect{\Wtil_i^j}=w_i^j$, this would still lead to an unbiased (albeit noisier) estimator of the likelihood. However, now we have a weight that is either $w^*_i$ or $0$, and it is natural to record the success as $\Stil_i^j=\Wtil_i^j/w^*_i\in \{0,1\}$. Then,
\[
\Expect{\Stil_i^j}=\frac{1}{w_i^*}\Expect{\Wtil_i^j}=\frac{w_i^j}{w_i^*}=\frac{f(y_{i}|x_{t_i}^j)}{\sup_{x\in \cX}f(y_i|x)}.
\]

\section{Further analysis of the case of partial observations}
\label{sec.more.partial}
We continue from Section \ref{sec.tune.partial}.

In the case of very large $\ess$, each $\zvec_t$, $t=1,\dots,T$ is, essentially, $\ess-1$ independent draws from the corresponding filtering distribution, which we denote by $g_t(z_t|x^*_{1:t})$. For moderately large $\ess-1$, we expect this to hold approximately, so that
\[
\begin{aligned}
\Expect{\prod_{t=1}^T \Pbar_t^2}
&\approx
\prod_{t=1}^T\Expect{\left(\frac{1}{\ess-1}\sum_{j=1}^{\ess-1}p_t(Z_t^j)\right)^2}
=
\prod_{t=1}^T\left(\Expect{p_t(Z_t)}^2+\frac{1}{\ess-1}\Var{p_t(Z_t)}\right)\\
&=
\left(\prod_{t=1}^T\Expect{p_t(Z_t)}\right)^2
\prod_{t=1}^T\left\{1+\frac{1}{\ess-1}\frac{\Var{p_t(Z_t)}}{\Expect{p_t(Z_t)}^2}\right\}\\
&\le
\left(\prod_{t=1}^T\Expect{p_t(Z_t)}\right)^2
\exp\left\{\frac{1}{\ess-1}\sum_{t=1}^T\frac{\Var{p_t(Z_t)}}{\Expect{p_t(Z_t)}^2}\right\}
\end{aligned},
\]
where, for each $t=1,\dots,T$, $Z_t,Z_t^1,\dots,Z_t^{\ess-1}$ are independent draws from $g_t(\cdot|x^*_{1:t})$.
Noting $\Expect{\prod_{t=1}^T \Pbar_t}\approx \left(\prod_{t=1}^T\Expect{p_t(Z_t)}\right)$ leads to an approximate bound on $\Vrel$ of
\[
\exp\left\{
\frac{T}{\ess-2}
+
\frac{1}{\ess-1}\sum_{t=1}^T\frac{\Var{p_t(Z_t)}}{\Expect{p_t(Z_t)}^2}
\right\}
-1.
\]
When compared with \eqref{eqn.VrelExact}, the additional term might be expected to dominate if the sum of the relative variances of the $p_t(Z_t)$ is larger than $T$.

\section{Further discussion}
\label{app.furtherDiscussion}
In the special case of partial but exact observations where at time $t$, $\ess$ is achieved after $M$ simulations, the transition probability must be estimated as $(\ess-1)/(M-1)$; however, the pool for resampling and then propagating can, in fact, consist of all $\ess$ successes, rather than the first $\ess-1$. This is because all $\ess$ of the particles' weights are $1$, so they are indistinguishable. This does not apply in the case of general weights and provides only a very minor increase in efficiency, so it was not mentioned when describing Algorithm \ref{FullFrank}.

Algorithms \ref{OneStepFrank} and \ref{FullFrank} could be further generalised to allow a cost $c^j$ to be associated with each simulation, and to stop if a certain cost budget, $\ecc$, has been exceeded. For example, in Algorithm \ref{OneStepFrank}, the condition on Line 4 would change to $m<m_+$ \textbf{and}  $\sum_{j=1}^{m-} s^j < \ess$ \textbf{and} $\sum_{j=1}^{m-}c^j< \ecc$ and the condition on Line 8 to $m=m_-$ \textbf{or} $\{\sum_{j=1}^{m-} s^j < \ess$ \textbf{and} $\sum_{j=1}^{m-}c^j< \ecc\}$. One could also set $m_+=\infty$, effectively replacing the maximum number of iterations with the cost budget. Alternatively, the algorithm could be generalised to simulate in batches of size $b$ throughout, where $b$ is the number of cores available and $m_-$ is a multiple of $b$.

The Frankenfilter Algorithm \ref{FullFrank} could also be applied within a correlated pseudo-marginal MCMC scheme \cite[]{DDP2018}. Furthermore, Algorithms \ref{BaseFrank} and \ref{OneStepFrank} could be employed within the augmented correlated pseudo-marginal approach of \cite{GolSher2022}, potentially reducing the required number of successes considerably. 

\subsection{Addendum to Section \ref{sec.TuneExact}}
\label{sec.TuneExactAddendum}
As an addendum to Section \ref{sec.TuneExact}, if we \emph{did} have a representative parameter value and we \emph{could} somehow estimate each $p_t$ for this parameter, then we could choose to try to minimise the expected number of simulations, $\sum_{t=1}^T\ess_t/p_t$, subject to obtaining a target $\Vrel$; \emph{i.e.}, subject to $\sum_{t=1}^T(1-p_t)/(\ess_t-2)=\log (1+\Vrel)$.  The method of Lagrange multipliers gives a solution of $\ess_t-2\propto \sqrt{p_t(1-p_t)}$. Of course, we could also use a standard particle filter, aiming to minimise the total number of simulations $\sum_{t=1}^T n_t$, subject to a fixed relative variance $\approx \exp\left\{\sum_{t=1}^T(1-p_t)/(n_tp_t)\right\}-1$, which suggests setting $n_t\propto \sqrt{(1-p_t)/p_t}$. These observation-specific sets of tunings are exactly the kind of user-unfriendly course of action that the Frankenfilter is designed to obviate, so we did not pursue this possibility further.

\end{document}